\documentclass[a4paper, 11pt,onecolumn]{article}
  \usepackage[margin=1in]{geometry}
\usepackage{amsmat h,amsthm,amsfonts}
\usepackage{amssymb}
\usepackage{graphicx}
\usepackage{booktabs}
\usepackage{setspace}
\usepackage{enumerate}
\newtheorem{thm}{Theorem}
\newtheorem{subthm}{Theorem}
\newtheorem{prop}{Proposition}
\newtheorem{coro}{Corollary}

\newcommand{\R}{{\mathbb{R}}}

\newtheorem{lem}{Lemma}

\newtheorem{examples}{Example}

\begin{document}
\author{Quanyan Zhu, Hamidou Tembine, Tamer Ba\c{s}ar\footnote{Q. Zhu and T. Ba\c{s}ar are with Coordinated Science Laboratory, University of Illinois at Urbana-Champaign, Urbana, IL, USA. (Email: \{zhu31, basar1\}@illinois.edu) H. Tembine is with Department of Telecommunications, \'Ecole  Sup\'erieure d'Electricit\'e (SUPELEC), France. (Email: tembine@ieee.org)}}
\title{Evolutionary Games for Multiple Access Control\footnote{The material in this paper was partially presented in \cite{TZB08a} and \cite{TZB08b}.}\ \footnote{This work was supported in part by a grant from AFOSR and by
a MURI grant.}}
\maketitle
\begin{abstract}

In this paper, we formulate an evolutionary multiple access control game with continuous-variable actions and coupled constraints. We characterize equilibria of the game and show that  the pure  equilibria are Pareto optimal and  also resilient to deviations by coalitions of any size, i.e., they are strong equilibria. We use  the concepts of  price of anarchy and strong price of anarchy to  study the performance of the system. The paper also addresses how to select one specific equilibrium solution using the concepts of normalized equilibrium and evolutionarily stable strategies. We examine the long-run behavior of these strategies under  several classes of evolutionary game dynamics, such as Brown-von Neumann-Nash dynamics, Smith dynamics and replicator dynamics. In addition, we examine correlated equilibrium for the single-receiver model. Correlated strategies are based on signaling structures before making decisions on rates. We then focus on evolutionary games for hybrid additive white Gaussian noise multiple access channel with multiple users and multiple receivers, where each user chooses a rate and splits it over the receivers. Users have coupled constraints determined by the capacity regions.  Building upon the static game, we formulate a system of hybrid evolutionary game dynamics using G-function dynamics and Smith dynamics on rate control and channel selection, respectively.  We show that the evolving game has an equilibrium and illustrate these dynamics with numerical examples.

\end{abstract}

\section{Introduction}

Recently, there has been much interest in understanding the behavior
of multiple access controls under constraints. Considerable amount of work has been carried out on the problem of how users can obtain an acceptable throughput by choosing  rates independently.
 Motivated by the interest in studying a large population of users playing the game over time, evolutionary game
theory was found to be an appropriate framework for communication networks. It has been applied to  problems such as power control in wireless networks  and mobile interference control \cite{networking,ieeesmc, belmega2010, cioffi2003}.

The game-theoretical models considered in the previous studies on user behaviors in CDMA, \cite{Alpcan01, Goodman02}, are static one-shot non-cooperative games in which users are assumed to be rational and optimize their payoffs independently. Evolutionary game theory, on the other hand, studies games that are played repeatedly, and focuses on the strategies that persist over time, yielding the best fitness of a user in a non-cooperative  environment on a large time scale.

In \cite{pairwise}, an additive white Gaussian noise (AWGN) multiple access channel problem was modeled
 as a noncooperative game with pairwise interactions,
in which users were modeled as rational entities whose only
interest was to maximize their own communication rates. The authors obtained the  Nash equilibrium of the two-user
game and introduced a two-player evolutionary game model with {\it pairwise interactions} based on replicator dynamics. However,  the case when interactions   are not pairwise arises frequently  in communication networks, such as the  Code Division Multiple Access (CDMA) or the Orthogonal Frequency-Division Multiple Access (OFDMA) in Worldwide Interoperability for Microwave Access (WiMAX) environment \cite{networking}.


In this  work, we extend the study of \cite{pairwise} to wireless communication systems with an arbitrary number of users corresponding to each receiver. We formulate a static non-cooperative game with $m$ users subject to rate capacity constraints,  and  extend the constrained game to a dynamic evolutionary game with a large number of users whose strategies evolve  over time.
Different from evolutionary games with discrete and finite number of actions, our model is based on a class of  continuous games, known as  {\it continuous-trait games}.  Evolutionary games with continuum action spaces can be encountered in a wide variety of applications
in evolutionary ecology, such as evolution of phenology, germination, nutrient
foraging in plants, and predator-prey foraging \cite{continuous,vincent}.

In addition to the single receiver model, we investigate the case with multiple users and receivers.
We first formulate a static hybrid non-cooperative game  with $N$ users who rationally make decisions
on the rates as well as the channel selection subject to rate capacity constraints of each receiver.  We extend
the static game to a dynamic evolutionary game by viewing rate selections governed by a fitness
function parameterized by the channel selections.  Such a concept of a hybrid model has appeared earlier in \cite{alpcan} and \cite{AlpcanBasar2004}, in the context of hybrid power control in CDMA systems. The strategies of channel selections determine the
long-term fitness of the rates chosen by each user. We formulate such dynamics based on generalized Smith
dynamics and generating fitness function (G-function) dynamics.

\subsection{Contribution}
The main contributions of this work can be summarized as follows.
We first introduce a game-theoretic framework for local interactions between many users and a single receiver.  We show that the static continuous-kernel rate allocation game with coupled rate constraints has a convex set of pure Nash equilibria, coinciding with the maximal face of the polyhedral capacity region.
All the pure equilibria are Pareto optimal and are also strong equilibria, resilient to simultaneous deviation by coalition of any size.
 We show that the pure Nash equilibria in the rate allocation problem are 100\% efficient in terms of   Price of Anarchy (PoA) and constrained Strong Price of Anarchy (CSPoA).
We study the stability of strong equilibria, normalized equilibria, and evolutionary stable strategies (ESS) using evolutionary game dynamics such as  Brown-von Neumann-Nash dynamics, generalized Smith dynamics, and replicator dynamics.
We further investigate the correlated equilibrium of the multiple access game where the receiver can send signals to the users to mediate the behaviors of the transmitters.

Based on the single-receiver model, we then propose an evolutionary game-theoretic framework for the hybrid additive white Gaussian noise multiple access channel. We consider a communication system of multiple users and multiple receivers, where each user chooses a rate and splits it over the receivers. Users have coupled constraints determined by the capacity regions. We characterize  Nash equilibrium of the static game and show the existence of the equilibrium under general conditions. Building upon the static game, we formulate a system of hybrid evolutionary game dynamics using G-function dynamics and Smith dynamics on rate control and channel selection, respectively.  We show that the evolving game has an equilibrium and illustrate these dynamics with numerical examples.

\subsection{Organization of the paper}
The rest of the paper is structured as follows. We present in Section \ref{localInteractions}  the evolutionary game model of rate allocation in additive white Gaussian multiple access wireless networks, and analyze its equilibria and Pareto optimality in Section \ref{payoffsSection}. In Section~\ref{secrobust}, we present strong equilibria and price of anarchy of the game. In Section \ref{secselection}, we discuss how to select one specific equilibrium such as normalized equilibrium and evolutionary stable strategies. Section~\ref{secdynamics} studies the stability of equilibria and evolution of strategies using game dynamics. Section~\ref{CESection} analyzes the correlated equilibrium of the multiple acess game.

In Section \ref{HybridRateModel}, we present the hybrid rate control model where users can choose the rates and the probability of  the channels to use.
In Section \ref{characterization}, we characterize the Nash equilibrium of the constrained hybrid rate control game model, pointing out the existence of the Nash equilibrium of the hybrid model and methods to find it. In
Section \ref{evolutionaryGames},  we apply evolutionary dynamics to both rates and channel selection probabilities. We use simulations to demonstrate the validity of these proposed dynamics and illustrate the evolution of the overall evolutionary dynamics of the hybrid model.
Section~\ref{secconclud} concludes the paper. For reader's convenience, we summarize the notations in Table \ref{notation} and the acronyms in Table \ref{acronyms}.


\begin{table}[t]
\caption{List of Notations}
\begin{center}
\begin{tabular}{ll@{}c@{}lc}
\toprule
  Symbol & Meaning  \\
\midrule
  $\mathcal{N}$  & set of $N$ users  \\
    ${\Omega}$  & a subset of $N$ users  \\
    $\mathcal{J}$  & set of $J$ receivers  \\
    $\mathcal{A}_i$ & action set of user $i$\\
  $P_i$ & maximum power of user $i$\\
    $h_i$ & channel gain of user $i$\\
   $\alpha_i$ & rate of user $i$ \\
     ${p}_{ij}$ & probability of user $i$ selecting receiver $j$\\
  $u_i$  & payoff of user $i$ \\
    $\overline{U}_i$  & expected payoff of user $i$ \\
      $\mathcal{C}$ & capacity region of a set $\mathcal{N}$ of users  in a single receiver case \\
  $\mathcal{C}(j)$ & capacity region of a set $\mathcal{N}$ of users  at receiver $j$ \\
    $\lambda_i$ & distribution over the action set $\mathcal{A}_i$\\ 
    $\mu$ & population state\\ \bottomrule
\end{tabular}
\end{center}
\label{notation}
\end{table}%

\begin{table}[t]
\caption{List of Acronyms}
\begin{center}
\begin{tabular}{ll@{}c@{}lc}
\toprule
  Abbreviation & Meaning  \\
\midrule
  AGWN  & Additive White Gaussian Noise  \\
  MAC  & Multiple Access Control  \\
    MISO  & Multi-Input and Multi-Output  \\
    CCE  & Constrained Correlated Equilibrium  \\
    ESS & Evolutionary Stable Equilibrium \\
  NE & Nash Equilibrium\\
    PoA & Price of Anarchy\\
      SPoA & Strong Price of Anarchy\\ \bottomrule
\end{tabular}
\end{center}
\label{acronyms}
\end{table}%

\section{ AWGN Mutiple Access Model: Single Receiver Case } \label{secmodel}

We consider a communication system consisting of several receivers and several senders (see Figure \ref{figfuncttt3}). At each time, there are several simultaneous local interactions (typically, at each receiver there is a local interaction). Each local interaction corresponds to a non-cooperative  one-shot game with common constraints. The opponents do not necessarily stay the same from a given time slot to the next one. Users revise their rates in view of their payoffs and the coupled constraints (for example by using an evolutionary process, a learning process or a trial-and-error updating process). The game evolves  over time.
Users are  interested in maximizing a fitness function based on
their own communication rates at each time, and  they are aware of the fact that
the other users have the same goal.
The coupled power and rate constraints are
also common knowledge. Users have to choose independently
their own coding rates at the beginning of the communication,
where the rates selected by a user may be either deterministic,
or chosen from some distribution. If the rate profile arrived at as a result of these independent decisions lies
in the capacity region, users will communicate at that operating
point. Otherwise, either the receiver is unable to decode any
signal and the observed rates are  zero, or only one of the
signals can be decoded. The latter  occurs when  all the other users are transmitting at
or below a safe rate. With these assumptions, we can define
a constrained non-cooperative game. The set of allowed strategies for user $i$ is the set
of all probability distributions over $[0,+\infty[,$ and the payoff is a function of
the  rates. In addition, the rational action (rate) sets are restricted to lie in the capacity regions (the payoff is zero if the constraint is violated).
In order to study the interactions between the selfish or partially cooperative users
and their stationary rates in the long run,
we propose to model the problem of rate allocation in Gaussian multiple access channels as an
evolutionary game with a continuous action space and coupled constraints. The development
of evolutionary game theory is a major contribution of
biology to competitive decision making and the evolution of cooperation. The key concepts of
evolutionary game theory are
 (i) {\it Evolutionary Stable States} \cite{smith},  which is a refinement of equilibria, and (ii) {\it Evolutionary Game Dynamics}  such as replicator dynamics \cite{taylor}, which describes
the evolution of strategies or frequencies of use of strategies in time, \cite{vincent,hofbauer}.
\begin{figure}[htb]
\centering
\includegraphics[scale=0.5]{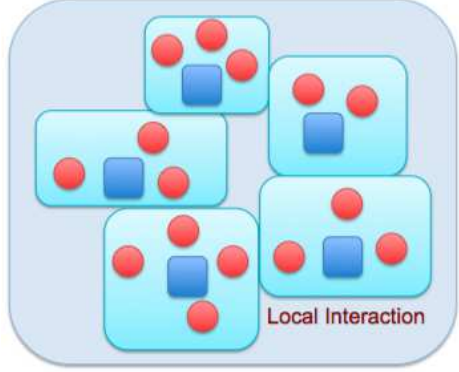}
\caption{A population: distributed receivers and senders, represented by blue rectangles and red circles respectively.} \label{figfuncttt3}
\end{figure}

The single population evolutionary rate allocation game is described as follows:
there is one population of senders (users) and several receivers. The number of senders is  large. At  each time, there are many one-shot games called {\it local interactions}.
Each sender of the
population chooses from his set of strategies
${\mathcal{A}_i}$ which is a  non-empty, convex and compact subset of $\mathbb{R}.$ Without loss of generality, we can suppose that user $i$ chooses its rate in the interval $\mathcal{A}_i=[0,C_{\{i\}}]$, where  $C_{\{i\}}$ is the rate upper bound for user $i$ (to be made precise shortly), as outside of the capacity region the payoff (as to be defined later) will be zero.
Let $\Delta({\mathcal{A}}_i) $ be
the set of probability distributions over the  pure strategy set $\mathcal{A}_i.$ The set
$\Delta({\mathcal{A}}_i)$ can be interpreted as the set of mixed strategies.
It is also interpreted as the set of distributions of strategies among
the population. Let $\lambda_i\in \Delta({\mathcal{A}}_i),$ and $E$ be a   $\lambda_i-$ measurable subset of $\mathbb{R}^N$; then  $\lambda_i(E)$ represents the fraction of users choosing
a strategy out of $ E$, at time $t.$ A distribution $\lambda_i\in \Delta({\mathcal{A}}_i) $ is sometimes called the ``state"
 of the population.   We denote by $\mathbb{B}(\mathcal{A}_i)$ the Borel $\sigma-$algebra on ${\mathcal{A}}_i$ and by $d(\lambda,\lambda')$ the distance between two states measured with the respect to the weak topology.
 Each user's payoff depends on opponents' behavior through
the distribution of opponents' choices and of their strategies. The payoff of a user $i$ in a local interaction with $(N-1)$ other users is given as a function $u_i:\ \mathbb{R}^N\longrightarrow \mathbb{R}.$ The rate profile $\alpha\in\mathbb{R}^N$ must belong to a common capacity region $\mathcal{C}\subset\mathbb{R}^N $ defined by $2^N-1$ linear inequalities. The
expected payoff of a sender $i$ transmitting at a rate $a$ when the state of the
population is $\mu\in \Delta(\mathcal{A}_i)$ is given by  $F_i(a,\mu).$  The expected payoff for user $i$ is 
$$F_i(\lambda_i,\mu):=\int_{\alpha\in \mathcal{C}}u_i(\alpha)\ \lambda_i(d\alpha_i) \prod_{j\neq i}\mu(d\alpha_j) .$$
The population state is subjected to the ``mixed extension" of capacity constraints $\mathcal{M}(\mathcal{C}).$
This will be discussed in Section \ref{secdynamics} and will be made more precise later.

\subsection{Local Interactions}\label{localInteractions}
Local interaction refers to the problem setting of one receiver and its uplink additive white Gaussian noise (AWGN) multiple access channel with $N$ senders with coupled constraints (or actions). The signal at the receiver is given by $ Y=\xi+\sum_{i=1}^{N}X_i$ where $X_i$ is a transmitted signal of user $i$ and $\xi$ is a zero-mean Gaussian noise with variance $\sigma_0^2.$ Each user has an individual power constraint $\mathbb{E}(X_i^2)\leq P_i$ and the channel gain $h_i$.  The optimal power allocation scheme
is to transmit at the maximum power available, i.e. $P_i$, for each user. Hence, we consider the case in which maximum power is attained. The decisions of the users  then  consist of choosing
their communication rates, and the receiver's  role is to
decode, if possible. The capacity region is the set of all vectors $\alpha\in \R^{N}_{+}$ such that users $i\in\mathcal{N}:=\{1,2,\ldots,N\}$  can reliably communicate at rate $\alpha_i, ~i\in\mathcal{N}.$
The capacity region $\mathcal{C}$ for this channel is the set
\begin{eqnarray} \label{tet}
\mathcal{C}=
\nonumber \left\{\alpha\in \R^{N}_{+} ~\bigg|~ \sum_{i\in \Omega}\alpha_i\leq   \log\left(1+|\Omega|\frac{P_ih_i}{\sigma^2_0}\right).
 \forall\ \emptyset \subsetneqq \Omega\subseteq \mathcal{N}\right\},
 \end{eqnarray}

\begin{examples}{(Example of capacity region with three users)}
In this example, we illustrate the capacity region with three users. Let $\alpha_1,\alpha_2,\alpha_3$ be the rates of the users and $P_i=P, h_i=h, \forall i=1,2,3$. Based on (\ref{tet}), we obtain a set of inequalities
$$\left\{\begin{array}{l}\alpha_1\geq 0,\alpha_2\geq 0,\alpha_3\geq 0\\
\alpha_i\leq \log\left(1+\frac{Ph}{\sigma_0^2}\right), i = 1, 2, 3\\
\alpha_i+\alpha_j\leq \log\left(1+2\frac{Ph}{\sigma_0^2}\right), i\not= j, i,j = 1,2,3. \\
\alpha_1+\alpha_2+\alpha_3\leq \log\left(1+3\frac{Ph}{\sigma_0^2}\right)\\
\end{array}\right., $$ or in the compact notation, $M_3\gamma_3\leq \zeta_3,\ $ where  $$ \gamma_3:=\left[\begin{array}{c}\alpha_1\\ \alpha_2\\ \alpha_3\end{array}\right]\in\mathbb{R}_+^3,\
\zeta_3:=\left[\begin{array}{c}C_{\{1\}}\\ C_{\{2\}}\\ C_{\{3\}}\\ C_{\{1,2\}}\\ C_{\{1,3\}}\\ C_{\{2,3\}}\\ C_{\{1,2,3\}}\end{array}\right] , M_3:=\left[\begin{array}{ccc}1 & 0 & 0 \\ 0 & 1 & 0\\ 0 & 0 & 1\\
 1 & 1 & 0\\ 1& 0 & 1\\ 0 & 1& 1\\ 1 & 1 & 1
 \end{array}\right]\in \mathbb{Z}^{7\times 3}.$$ Note that $M_3$ is a totally unimodular matrix.
By letting $Ph=25,\sigma_0^2=0.1,$ we sketch in Figure~\ref{figfunctt3} the capacity region with three  users.
\begin{figure}[htb]
\centering
\includegraphics[scale=0.5]{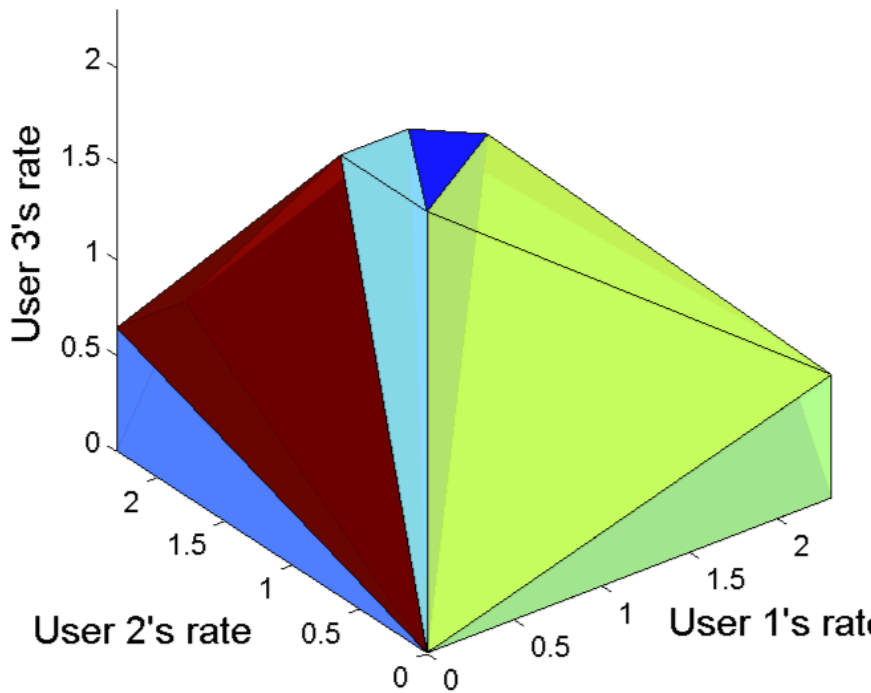}\\ 
\caption{Capacity region with three users.}
  \label{figfunctt3}
\end{figure}
\end{examples}

 The capacity region reveals a competitive nature of the interactions among senders: if a user $i$ wants to communicate at a higher rate, one of the other users  has to lower his rate;  otherwise, the capacity constraint is violated.  We let $$r_{i,\Omega} :=\log\left(1+\frac{P_ih_i}{\sigma_0^2+\sum_{i'\in \Omega,i'\neq i} P_{i'}h_{i'}}\right), i\in\mathcal{N}, \Omega\subseteq\mathcal{N}$$  denote the  bound on the rate  of a user when the signals of the $|\Omega|-1$ other users are treated as  noise.

Due to the noncooperative nature of the rate allocation, we can formulate the one-shot game $$\Xi =\langle \mathcal{N}, (\mathcal{A}_i)_{i\in\mathcal{N}}, (u_i)_{i\in\mathcal{N}}\rangle\,,$$ where the set of users $\mathcal{N}$ is the set of players, $\mathcal{A}_i$, $i\in  \mathcal{N}$, is  the set of actions, and $u_i$, $i\in  \mathcal{N}$, are the payoff functions.

\subsection{Payoffs}\label{payoffsSection}
We define $u_i:\prod_{i=1}^N\mathcal{A}_i\rightarrow \mathbb{R}_+$ as follows:
\begin{eqnarray}
u_i(\alpha_{i},\alpha_{-i})&=&\upharpoonleft_{\mathcal{C}}(\alpha)g_i(\alpha_i,\alpha_{-i})\\ &=&\left\{ \begin{array}{ll} g_i(\alpha_i) & {\textrm{~if~}}\   (\alpha_{i},\alpha_{-i})\in \mathcal{C}\\
0 & \mbox{otherwise}
\end{array}
\right.,
\end{eqnarray}
 where  $\upharpoonleft_{\mathcal{C}}$ is the indicator function; $\alpha_{-i}$ is a vector consisting of other players' rates, i.e., $\alpha_{-i}=[\alpha_1,\ldots,\alpha_{i-1},\alpha_{i+1},\ldots,\alpha_N]$  and $u_i$
 is a positive and strictly increasing  function for each fixed $\alpha_{-i}$. Since the game is subject to coupled constraints, the action set $\mathcal{A}_i$ is coupled and dependent on other players' actions. Given the strategy profile $\alpha_{-i}$ of other players, the constrained action set $\mathcal{A}_i$ is given by
 \begin{equation}
 \mathcal{A}_i(\alpha_{-i}):=\{\alpha_i\in [0,C_{\{i\}}],\ (\alpha_i,\alpha_{-i})\in \mathcal{C} \}
 \end{equation}  We then have an asymmetric game. The minimum rate that the user $i$ can guarantee in the feasible regions is $r_{i,\mathcal{N}}$ which is different than $r_{j,\mathcal{N}}.$

 Each user $i$ maximizes $u_i(\alpha_{i},\alpha_{-i})$ over the coupled constraint set.
Owing to the monotonicity of the function $g_i$ and the inequalities that define the capacity region, we obtain the following lemma.
\begin{lem}
 Let $\overline{BR}_i(\alpha_{-i})$ be the best reply to the strategy $\alpha_{-i}$, defined by  $$ \overline{BR}^i(\alpha_{-i})=\arg\max_{y\in \mathcal{A}_i(\alpha^{-i})}u_i(y,\alpha_{-i}).$$ $\overline{BR}_i$ is a non-empty single-valued correspondence (i.e. a  standard function), and is given by
\begin{equation}
\max \left(r_{i,\mathcal{N}},\min_{\Omega\in\Gamma_i}\left\{C_{\Omega}-\sum_{k\in\Omega\backslash\{i\}}\alpha_k\right\} \right),
\end{equation} \label{lembr1}
where $\Gamma_i=\{\Omega\in 2^{\mathcal{N}}, i\in\Omega\}$.
\end{lem}
\begin{prop}\label{ne}
 The set of Nash equilibria is $$\left\{ (\alpha_i,\alpha_{-i})\ |\  \alpha^{i}\geq r_{i,\mathcal{N}},\sum_{i\in\mathcal{N}}\alpha_{i}=C_{\mathcal{N}}\right\}.$$
All these equilibria are optimal in Pareto sense.
\end{prop}
\begin{proof}
Let $\beta$ be a feasible solution, i.e., $\beta\in\mathcal{C}.$ If $\sum_{i=1}^N \beta_i< C_{\mathcal{N}}=\log\left(1+\sum_{i\in \mathcal{N}} \frac{P_ih_i}{\sigma_0^2}\right),$ then at least one of the users can improve its rate (hence its payoff) to reach one of the faces of the capacity region. We now check the strategy profile on the face
$$ \left\{ (\alpha_i,\alpha_{-i})\ \bigg|\  \alpha^{i}\geq r_{i,\mathcal{N}},\sum_{i=1}^N\alpha_{i}=C_{\mathcal{N}}\right\}.$$
If  $\beta\in \left\{ (\alpha_i,\alpha_{-i})\ \bigg|\  \alpha_{i}\geq r_{i,\mathcal{N}},\sum_{i=1}^N\alpha_{i}=C_{\Omega}\right\},$ then
from the Lemma \ref{lembr1}, $\overline{BR}_i(\beta_{-i})=\{\beta_i\}.$ Hence, $\beta$ is a strict equilibrium. Moreover, this strategy $\beta$ is Pareto optimal because  the rate of each user is maximized under the capacity constraint. These strategies are social welfare optimal if the total utility $\sum_{i=1}^N u_i(\alpha_i,\alpha_{-i})=\sum_{i=1}^N g_i(\alpha_i)$ is maximized subject to constraints.
\end{proof}

Note that the set of pure Nash equilibria is a convex subset of the capacity region.

\subsection{Robust Equilibria and Efficiency Measures} \label{secrobust}

\subsubsection{Constrained Strong Equilibria and Coalition Proofness}

An action profile  in a local interaction between $N$ senders is a constrained $k-$strong equilibrium if it is feasible and no coalition of size $k$ can
improve the rate transmissions of each of its members with  respect to the capacity constraints. An action is a constrained strong equilibrium \cite{aumann} if it is a constrained $k-$strong equilibrium for any size $k.$  A strong equilibrium is then a policy from which no coalition
(of any size) can deviate and improve the transmission rate of every
member of the coalition (group of the simultaneous moves), while possibly lowering the transmission rate
of users outside the coalition group. This notion of  constrained strong equilibria\footnote{Note that the set of constrained strong equilibria  is a subset of the set of Nash equilibria (by taking coalitions of size one) and any constrained strong equilibrium is Pareto optimal (by taking coalition  of full size).} is very attractive because it is resilient against coalitions of users. Most of the games do not admit any strong
equilibrium but in our case we will show that the multiple access channel game has several strong equilibria.

\begin{subthm} Any rate profile on the maximal  face of the capacity region $\mathcal{C}:$
$$ Face_{\max}(\mathcal{C}):=\left\{ (\alpha_i,\alpha_{-i})\in\mathbb{R}^N\ |\  \alpha_{i}\geq r_N,\sum_{i=1}^N\alpha_{i}=C_{\mathcal{N}}\right\},$$
is a constrained strong equilibrium.
\end{subthm}
\begin{proof}
 We remark that if the rate profile $\alpha$ is not on the maximal face of the capacity region, then $\alpha$ is not resilient to deviation by a single user. Hence, $\alpha$ cannot be a constrained strong equilibrium. This says that a necessary condition for a rate profile to be a strong equilibrium is to be in the subset  $Face_{\max}(\mathcal{C}).$ We now prove that the condition: $\alpha\in Face_{\max}(\mathcal{C})$ is sufficient. Let $\alpha\in Face_{\max}(\mathcal{C}).$ Suppose that $k$ users deviate simultaneously from the rate profile $\alpha.$ Denote by $Dev$ the set of users which deviate simultaneously (eventually by forming a coalition). The rate constraints of the  deviants are
 \begin{enumerate}
\item ${\alpha}'_i\geq 0, \ \forall i\in Dev,$
\item $\sum_{i\in\bar{\Omega}}{\alpha}'_i\leq C_{\bar{\Omega}},\ \forall \bar{\Omega}\subseteq Dev,$
\item $\sum_{i\in \Omega\cap Dev}{\alpha}'_i\leq C_{\Omega}-\sum_{i\in \Omega, i\notin Dev}\alpha_i$, $\ \forall {\Omega}\subseteq \mathcal{N},\ \Omega\cap Dev\neq \emptyset.$
\end{enumerate}
In particular, for $\Omega=\mathcal{N},$ we have $\sum_{i\in Dev}{\alpha'}_i\leq C_{\mathcal{N}}-\sum_{ i\notin Dev}\alpha_i.$ The total rate of the deviants is bounded by
$C_{\mathcal{N}}-\sum_{ i\notin Dev}\alpha_i$, which is not controlled by the deviants. The deviants move to $({\alpha}'_i)_{i\in Dev}$ with  $$\sum_{i\in Dev}{\alpha}'_i <C_{\mathcal{N}}-\sum_{ i\notin Dev}\alpha_i\,.$$ Then, there  exists $i$ such that $\alpha_i>{\alpha}'_i.$  Since $g_i$ is non-decreasing, this implies that $g_i(\alpha_i)>g_i({\alpha}'_i).$ The user $i$ who is a member of the coalition $Dev$ does not improve its payoff.
If the rates of some of the deviants are increased, then the rates of some other users from coalition must decrease. If $({\alpha}'_i)_{i\in Dev}$ satisfies $$\sum_{i\in Dev}{\alpha}'_i =C_{\mathcal{N}}-\sum_{ i\notin Dev}\alpha_i\,,$$ then some users in the coalition $Dev$ have increased their rates compared with $(\alpha_i)_{i\in Dev}$ and some others in $Dev$ have decreased their rates of transmission (because the total rate is the constant $C_{\mathcal{N}}-\sum_{ i\notin Dev}\alpha_i).$ The users in $Dev$ with a lower rate ${\alpha}'_i\leq \alpha_i$ do not benefit by being a member of the coalition (Shapley criterion of membership of coalition does not hold) . And this holds for any $\emptyset\subsetneqq Dev\subseteqq\mathcal{N}.$ This completes the proof.
\end{proof}
\begin{coro} In the constrained rate allocation game, Nash equilibria and strong equilibria in pure strategies coincide.
\end{coro}

\subsubsection{Constrained Potential Function for Local Interaction}
Introduce the following function:
$$V(\alpha)=\upharpoonleft_{\mathcal{C}}(\alpha) \sum_{i=1}^N g_i(\alpha_i)\,,$$ where $\upharpoonleft_{\mathcal{C}}$ is the indicator function of $\mathcal{C}, i.e.,\ $ $\upharpoonleft_{\mathcal{C}}(\alpha)=1$ if $\alpha\in\mathcal{C}$ and $0$ otherwise.
 The function $V$ satisfies $$ V(\alpha)-V(\beta_i,\alpha_{-i})=g_i(\alpha_i)-g_i(\beta_i),\ \forall \alpha,(\beta_i,\alpha_{-i})\in\mathcal{C} .$$ If $g_i$ is  differentiable, then one has  $$\frac{\partial}{\partial \alpha_i}V(\alpha)= g'_i(\alpha_i)=\frac{\partial}{\partial \alpha_i}u_i$$ in the interior of the capacity region $\mathcal{C}$, and $V$ is a constrained potential function \cite{Zhu08} in pure strategies.

\begin{coro}
The local maximizers of $V$ in $\mathcal{C}$ are pure Nash equilibria. Global  maximizers of $V$ in $\mathcal{C}$ are both constrained strong  equilibria and social optima for the local interaction.
\end{coro}

\subsubsection{Strong Price of Anarchy}
Throughout  this subsection, we assume that the functions $g_i$ are the identity function, i.e.,  $g_i(x)=id(x):=x.$
One metric used to measure how much the performance of decentralized systems is
affected by the selfish behavior of its components is the {\it price of anarchy}. We present a similar price for strong equilibria under the coupled rate constraints. This notion of Price of Anarchy can be seen as an {\it efficiency metric} that measures the {\it price of selfishness} or decentralization and has been extensively used in the context of congestion games or routing games where typically users have to minimize a cost function \cite{BasarZhu2010,Roughgarden2005}. In the context of rate allocation in the multiple access channel, we define an equivalent  measure of price of anarchy  for  rate maximization problems.
One of the advantages of a strong equilibrium
is that it has the potential to reduce the distance between the
optimal solution and the solution obtained as an outcome
of selfish behavior, typically  in the case where the capacity constraint is violated at each time.
Since the constrained rate allocation game has  strong equilibria, we can define the  strong price of anarchy, introduced in \cite{stro}, as the ratio between the payoff of the worst constrained
strong equilibrium and  the  social optimum value which $C_{\mathcal{N}}$.
\begin{subthm}
The strong price of anarchy  of the  constrained rate allocation game is 1 for $g_i(x)=x.$
\end{subthm}
Note that
for $g_i\neq id,$ the CSPoA can be less than one. However, the  optimistic
price of anarchy of the {\it best constrained  equilibrium}, also called {\it price of stability} \cite{psta},  is one for any function $g_i$ i.e the  efficiency of ``best" equilibria is $100\%.$

\subsection{Selection of Pure Equilibria} \label{secselection}
We have shown in the previous sections that our rate allocation game has a continuum of pure Nash equilibria and strong equilibria.  We
address now the problem of selecting one equilibrium which has certain
desirable properties: the normalized pure Nash equilibrium, introduced in \cite{rosen}; see also \cite{corre,cor,ponstein}. We introduce the
problem of constrained maximization
faced by each user  when the other rates are at the maximal face of the polytope $\mathcal{C}$:
\begin{eqnarray}
\max_{\alpha} & u_i(\alpha)\\
 \textrm{s.t. }&\alpha_1+\ldots+\alpha_N=C_{\mathcal{N}}
\end{eqnarray}
 for which the corresponding Lagrangian for user $i$ is given by $$L_i(\alpha,\zeta)=u_i(\alpha)-\zeta_i\left(\sum_{i=1}^N\alpha_i-C_{\mathcal{N}}\right).$$ From Karush-Kuhn-Tucker optimality conditions, it follows that there exists $\zeta\in\mathbb{R}^N$ such that $$ g_i'(\alpha_i)=\zeta_i,\ \sum_{i=1}^N\alpha_i=C_{\mathcal{N}}.$$ For a fixed vector $\zeta$ with identical entries, define the normal form game $\Gamma({\zeta})$ with $N$ users, where actions are taken as rates and the payoffs  given by $L(\alpha,\zeta).$ A normalized equilibrium is an equilibrium of the game $\Gamma(\zeta^*)$ where $\zeta^*$ is normalized into the form ${\zeta^*_i}=\frac{c}{\tau_i},\ c>0,\tau_i>0.$
We now have the following result due to Goodman \cite{cor} which implies Rosen's condition on uniqueness for strict concave games.
\begin{subthm} \label{ret1}
Let $u_i$ be a smooth and strictly concave function in $\alpha_i,$ each $u_i$ be convex in $\alpha_{-i}$, and there exist some $\zeta$ such that the weighted non-negative sum of the payoffs $
\sum_{i=1}^N\zeta_i u_i(\alpha)
$ is concave in $\alpha.$ Then, the matrix $G(\alpha,\zeta)+G^{T}(\alpha,\zeta)$ is  negative definite (which implies uniqueness) where $G(\alpha,\zeta)$ is the Jacobian with respect to $\alpha$ of
 $$h(\alpha,\zeta):=\left[\zeta_1 \nabla_1 u_1(\alpha), \zeta_2 \nabla_2 u_2(\alpha),\ldots,
 \zeta_{N} \nabla_{N} u_{N}(\alpha)
 \right]^T$$
and $G^{T}$ is the transpose of the matrix $G.$
\end{subthm}
This now leads to the following corollary for our problem.
\begin{coro}
If $g_i$ are non-decreasing strictly concave functions, then the rate allocation game has a unique normalized equilibrium which corresponds to an equilibrium of the normal form game with payoff $L(\alpha,\zeta^*)$ for some $\zeta^*.$
\end{coro}

\subsection{Stability and Dynamics} \label{secdynamics}
In this subsection, we study the stability of equilibria and several classes of evolutionary game dynamics under a symmetric case, i.e., $P_i=P, h_i=h, g_i=g, \mathcal{A}_i=\mathcal{A}, \  \forall i\in\mathcal{N}$. We will drop subscript index $i$ where appropriate. We show that the associated evolutionary game has a unique pure constrained evolutionary stable strategy.
\begin{prop}
The collection of rates $\alpha=\left(\frac{C_{\mathcal{N}}}{N},\ldots,\frac{C_{\mathcal{N}}}{N}\right)\,,$ i.e. the distribution of Dirac concentrated on the rate $\frac{C_{\mathcal{N}}}{N},$ is the unique  symmetric pure Nash equilibrium.
\end{prop}
\begin{proof}
Since the constrained rate allocation game is symmetric, there exists a symmetric  (pure or mixed) Nash equilibrium. If such an equilibrium exists in pure strategies, each user transmits with the same rate $r^*.$ It follows from Proposition \ref{ne} of Section \ref{payoffsSection}, and the bound $r_N\leq \frac{C_{\mathcal{N}}}{N}$ that $ r^*$ satisfies $N r^*=C_{\mathcal{N}}$  and $r^*$ is feasible.
\end{proof}

Since the set of feasible actions is convex, we can define convex combination of rates in the set of the feasible rates. For example, $\epsilon \alpha'+(1-\epsilon)\alpha$ is a feasible rate if $\alpha'$ and $\alpha$ are feasible. The symmetric rate profile  $(r, r,\ldots,r)$ is feasible if and only if $0\leq r\leq r^*=\frac{C_{\mathcal{N}}}{N}.$  We say that the rate $r$ is a constrained evolutionarily stable strategy (ESS) if it is feasible and for every {\it mutant strategy} $mut\neq \alpha$ there exists $\epsilon_{mut}>0$ such that
$$\left\{\begin{array}{cc}
r_{\epsilon}:=\epsilon\ mut+(1-\epsilon)r\in \mathcal{C} & \forall \epsilon\in(0,\epsilon_{mut})\\
u(r,r_{\epsilon},\ldots,r_{\epsilon})>u(mut, r_{\epsilon},\ldots,r_{\epsilon}) & \forall \epsilon\in(0,\epsilon_{mut})
\end{array}\right.$$

\begin{thm} The pure strategy $r^*=\frac{C_{\mathcal{N}}}{N}$ is a constrained evolutionary stable strategy.
\end{thm}
\begin{proof} Let $mut\leq r^*$
The rate $\epsilon\ mut+(1-\epsilon)r^*$ is feasible implies that $mut\leq r^*$ (because $r^*$ is the maximum symmetric rate achievable). Since $mut\neq r^*,$ $mut$ is strictly lower than $r^*.$ By monotonicity of the function $g,$ one has $ u(r^*,\epsilon\ mut+(1-\epsilon)r^*)>u(mut,\epsilon\ mut+(1-\epsilon)r^*),\ \forall \epsilon.$ This completes the proof.
\end{proof}

\subsubsection{Symmetric Mixed Strategies}
Define the mixed capacity region $\mathcal{M}(\mathcal{C})$ as the set of measures profile $(\mu_1,\mu_2,\ldots,\mu_N)$ such that $$\int_{\mathbb{R}_{+}^{|\Omega|}}\left(\sum_{i\in \Omega}\alpha_i\right)\prod_{i\in \Omega}\mu_i(d\alpha_i) \leq C_{\Omega},\ \forall \Omega\subseteq 2^{\mathcal{N}}.$$
Then, the payoff of the action $a\in\mathbb{R}_{+}$ satisfying $(a,\lambda,\ldots,\lambda)\in \mathcal{M}(\mathcal{C})$  can be defined as 
\begin{equation}\label{payoffFunctionF} 
F(a,\mu)=\int_{[0,\infty[^{N-1}} u(a,b_2,\ldots, b_N)\ \nu_{N-1}(db)\,,
\end{equation}
where $\nu_k=\bigotimes_{1}^{k} \mu$ is the product measure on $[0,\infty[^{k}.$
The constraint set becomes the set of probability measures on $\mathbb{R}_+$ such that  $$ 0\leq \mathbb{E}(\mu):=\int_{\mathbb{R}_+}\ \alpha_i\ \mu(d\alpha_i)\leq \frac{C_{\mathcal{N}}}{N}<C_{\{1\}}\,.$$

\begin{lem}
The payoff can be obtained as follows:
 $$F(a,\mu)=\upharpoonleft_{[0,C_{\mathcal{N}}-(N-1)\mathbb{E}(\mu)]} \times g(a)\times  \int_{b\in \mathcal{D}_a} \ \nu_{N-1}(db)=\upharpoonleft_{[0,C_{\mathcal{N}}-(m-1)\mathbb{E}(\mu)]}\times g(a)\nu_{N-1}(\mathcal{D}_a),$$
where $ \mathcal{D}_a=\{(b_2,\ldots, b_N)\ |\  (a,b_2,\ldots, b_N)\in \mathcal{C} \}\,.$
\end{lem}
\begin{proof} If the rate does not satisfy the capacity constraints, then the payoff is $0.$ Hence the {\it rational} rate for user $i$ is lower than $C_{\{i\}}.$ Fix a rate $a\in[0,C_{\{i\}}].$
Let $D^a_\Omega:=C_\Omega-a\delta_{\{1\in \Omega\}}.$ Then, a necessary condition to have a non-zero payoff is  $(b_2,\ldots, b_N) \in \mathcal{D}_a\,,$ where $\mathcal{D}_a=\{(b_2,\ldots, b_N)\in \mathbb{R}_{+}^{N-1},\ \sum_{i\in \Omega,i\neq 1}b_i\leq D^a_\Omega,\ \Omega\subseteq 2^{\mathcal{N}} \}.$
Thus, we have
\begin{eqnarray} \nonumber F(a,\mu)&=&\int_{\mathbb{R}_{+}^{N-1}} u(a,b_2,\ldots, b_N)\ \nu_{N-1}(db)\\ \nonumber & =&\int_{b\in \mathbb{R}_{+}^{N-1},\ (a,b)\in\mathcal{C}} g(a)\ \nu_{N-1}(db)\\
\nonumber &=& \upharpoonleft_{[0,C_{\mathcal{N}}-(N-1)\mathbb{E}(\mu)]} g(a) 
\times \int_{b\in \mathcal{D}_a} \ \nu_{N-1}(db)
\end{eqnarray}
\end{proof}

   \subsubsection{Constrained Evolutionary Game Dynamics}
   The class of evolutionary games in large population provides a simple
framework for describing strategic interactions among large numbers of users.  In this subsection we  turn to modeling the behavior of the users who play them.
Traditionally, predictions of behavior in game theory are based on some notion of equilibrium,
typically Cournot equilibrium, Bertrand equilibrium, Nash equilibrium, Stackelberg solution, Wardrop equilibrium or some refinement thereof. These notions require the assumption of equilibrium knowledge, which posits that each user correctly anticipates
how his opponents will act. The equilibrium knowledge assumption is too strong and is difficult
to justify in particular in contexts with large numbers of users.
As an alternative to the equilibrium approach, we propose an explicitly dynamic
 updating choice, a procedure in which users myopically update their behavior in response to
their current strategic environment. This dynamic procedure does not assume the automatic
coordination of users' actions and beliefs, and it can derive many specifications of users'
choice procedures.
These procedures are specified formally by defining a revision of rates called {\it revision
protocol} \cite{wbillbook}.  A revision
protocol takes current payoffs  and current mean rate and maps to conditional switch rates which describe how frequently users in some class playing rate $\alpha$  who
are considering switching rates switch to strategy $\alpha'.$ Revision protocols are flexible enough to
incorporate a wide variety of paradigms, including ones based on imitation, adaptation, learning, optimization, etc.

 We use here a class of continuous evolutionary dynamics.  We refer to \cite{wiopt,gamecomm,mass} for evolutionary game dynamics with or  without time delays. The continuous-time evolutionary game
dynamics on the measure space $(\mathcal{A}, \mathcal{B}(\mathcal{A}),\mu)$ is  given by
 \begin{equation} \dot{\lambda}_t(E)= \int_{a\in E}V(a,\lambda_t) \mu(da) \end{equation}
where $$V(a,\lambda_t)=K\left[\int_{x\in \mathcal{A}}\beta^x_a(\lambda_t)\lambda_t(dx)-\int_{x\in \mathcal{A}}\beta^a_x(\lambda_t) \lambda_t(dx)\right],$$ and $\beta^x_a$ represents the rate of mutation from $x$ to $a,$ and $K$ is a growth parameter.  $\beta^x_a(\lambda_t)=0$ if $(x,\lambda_t)$ or $(a,\lambda_t)$ is not in the (mixed) capacity region, $E$ is a $\mu-$measurable subset of $ \mathcal{A}.$
At each time $t,$ probability measure $\lambda_t$ satisfies
$\frac{d}{dt}\lambda_t(\mathcal{A})=0$.
\medskip
We examine the following classes of evolutionary game dynamics, namely, Brown-von Neumann-Nash dynamics, Smith dynamics and replicator dynamics, where $F(a, \lambda_t)$ is the payoff in (\ref{payoffFunctionF}) as defined in the previous subsection.
\begin{enumerate}[RD-1:]
\item{ Constrained Brown-von Neumann-Nash dynamics.}
$$ \beta^x_a(\lambda_t)=\left\{\begin{array}{cl} \max(F(a,\lambda_t)-\int_x F(x,\lambda_t)\ dx, 0) & \mbox{if}\ (a,\lambda_t),\ (x,\lambda_t)\in \mathcal{M}(\mathcal{C}), \\ 0 \ & \mbox{otherwise.}
\end{array}\right.
$$
%
%
\item{ Constrained Replicator Dynamics.}
$$\beta^x_a(\lambda_t)=\left\{\begin{array}{cl} \max(F(a,\lambda_t)-F(x,\lambda_t), 0) & \mbox{if}\ (a,\lambda_t),\ (x,\lambda_t)\in \mathcal{M}(\mathcal{C})\\ 0 \ & \mbox{otherwise.}\end{array}\right.$$

\item{ Constrained $\theta-$Smith Dynamics.}
$$ \beta^x_a(\lambda_t)=\left\{\begin{array}{cl} \max(F(a,\lambda_t)-F(x,\lambda_t), 0)^{\theta}  &\mbox{if}\ (a,\lambda_t),\ (x,\lambda_t)\in \mathcal{M}(\mathcal{C})\\ 0 \ &\mbox{otherwise.}\end{array}\right.,\ \theta\geq 1$$
\medskip
\end{enumerate}
 A common property that applies to all these dynamics is that  the set of Nash equilibria is a subset of rest points (stationary points) of the evolutionary game dynamics. Here we extend the concepts of these dynamics to evolutionary games with a continuum action space and coupled constraints, and more than two-users interactions. The counterparts of these results in discrete action space can be found in \cite{hofbauer,wbillbook}.

\begin{thm} Any  Nash equilibrium of the game is a rest point of the following evolutionary game dynamics: constrained Brown-von Neumann-Nash, generalized Smith dynamics, and replicator dynamics. In particular, the evolutionary stable strategies set  is a subset of the rest points of these constrained evolutionary game dynamics.
\end{thm}

\begin{proof}
It is clear for pure equilibria by using the revision protocols $\beta$ of these dynamics. Let $\lambda$ be an equilibrium. For any rate $a$ in the support of $\lambda,$ $\beta^a_x=0$ if $F(x,\lambda)\leq F(a,\lambda).$ Thus, if $\lambda$ is an equilibrium, the difference between the microscopic inflow and outflow is $V(a,\lambda)=0$, given that $a$ is the support of the measure $\lambda.$
\end{proof}

Let $\lambda$ be a  finite Borel measure on $[0,C_{\{i\}}]$ with full support. Suppose  $g$ is continuous on $[0,C_{\{i\}}].$ Then,
$\lambda$ is a rest point of the BNN dynamics if and only if $\lambda$ is a symmetric
Nash equilibrium.
Note that the choice of topology is an important issue when defining dynamics
convergence and stability of the dynamics. The most used in this area is the topology of the weak
convergence to measure closeness of two states of the system. Different
distances (Prohorov metric, metric on bounded and Lipschitz continuous
functions on $\mathcal{A}$) have been proposed. We refer the reader
to \cite{pierre}, and the references therein for more
details on {\it evolutionary robust strategy} and stability notions.

\subsection{Correlated Equilibrium}\label{CESection}

 In this subsection,  we analyze constrained correlated equilibria of multiple access (MAC) games. Building upon the signaling in the one-shot game, we formulate a system of evolutionary MAC games with evolutionary evolutionary game dynamics that describe the evolution of signaling, beliefs, rate control and channel selection, respectively.


We focus on correlated equilibrium in the single-receiver case. Correlated strategies are based on signaling structures before making decisions on rates. Different scenarios (with or without mediator, virtual mediator, cryptographic multi-stage signaling structure) have been proposed  in the literature \cite{forges,DHR00, forges2,forges3}.

\begin{figure}[htb]
\begin{center}
  \includegraphics[scale=0.5]{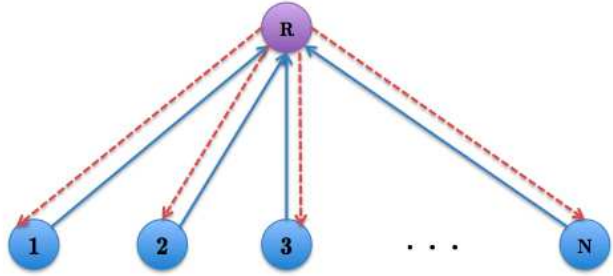}
  \caption{Signaling between multiple senders and a receiver}
  \label{CE}
\end{center}
\end{figure}

 In Figure \ref{CE}, we illustrate the signaling between multiple transmitters and one receiver. The receiver can act as a signaling device to mediate the behaviors of the transmitters. The correlated equilibrium has a strong connection with cryptography in that the private signal sent to the users can be realized by the coding and decoding in the network \cite{DHR00}. 

Let $\mathcal{B}$ be the set of  signals $\beta=[\beta_i,\beta_{-i}] \in  \mathbb{R}^{N}.$  The  values $\beta$ from the set of signals need to be in the feasible set $\mathcal{C}\subset \mathbb{R}^{N}$. Let $\mu\in \Delta{\mathcal{B}}$ be a probability measure over the set $\mathcal{B}.$ A constrained correlated equilibrium (CCE) $\mu^*$ need to satisfy the following set of inequalities,
\begin{eqnarray}
\int d\mu^*(\beta_i,\beta_{-i})\left[u_i(\alpha_i, \alpha_{-i}, \mid \beta_i)-u_i(\alpha_i', \alpha_{-i} \mid \beta_i)\right]\geq 0,  \forall i\in\mathcal{N}, \alpha'_i\in\mathcal{A}_i(\alpha_{-i}).\nonumber
\end{eqnarray}
Define a rule of assignment of user $i$ as a map its signals to its  action's set $\bar{r}_i:\ \beta_i \ \longmapsto \alpha_i.$ A CCE is then characterized by
\begin{eqnarray}
\int d\mu^*(\beta)\left[u_i(\alpha_i, \alpha_{-i}  \mid \beta_i)-u_i(r_i(\beta_i), \alpha_{-i} )\right]\geq 0, \forall i\in\mathcal{N}, \forall r_i \ \mbox{such that}\  \bar{r}_i(.)\in\mathcal{A}_i(\alpha_{-i}).
\end{eqnarray}

\begin{thm} The set of constrained pure Nash equilibria of the MISO game is given by
$$
\mbox{max-face}(\mathcal{C})=\left\{(\alpha_1,\ldots,\alpha_N)\ | \ \alpha_i\geq 0,\ \sum_{k\in\mathcal{N}}\alpha_k=C_{\mathcal{N}}\right\}$$
\end{thm}

We can characterize the CCE  using the above results as follows. 
\begin{lem}
Any mixture of constrained pure Nash equilibria of the MISO game  is a constrained correlated equilibrium.
\end{lem}
Note that the set of constrained correlated equilibria is bigger than the set of constrained Nash equilibria. For example, in a two-user case, the distribution $\frac{1}{2}\delta_{(r_1,C_{\{1,2\}}-r_1)}+\frac{1}{2}\delta_{(C_{\{1,2\}}-r_2,r_2)}$ is different than the Dirac distribution $\delta_{(\frac{r_1+C_{\{1,2\}}-r_2}{2},\frac{r_2+C_{\{1,2\}}-r_1}{2})}.$

\begin{proof}
Let $\mu$ be a probability distribution over some constrained pure equilibria. Then, $\mu\in \Delta(\mbox{max-face}(\mathcal{C})).$
For any $\beta$ such that $\mu(\beta)\neq 0,$ one has $$  \left[u_i(\alpha_i, \alpha_{-i}  \mid \beta_i)-u_i(\bar{r}_i(\beta_i), \alpha_{-i} )\right]\geq 0$$ for any measurable function $\bar{r}_i:\ [0,C_{\{i\}}]\longrightarrow [0,C_{\{i\}}].$ Thus, $\mu$ is a constrained correlated equilibrium.

\end{proof}

\begin{coro}
Any convex combination  of extreme point of the convex compact set $$
\mbox{max-face}(\mathcal{C})=\left\{ \alpha=(\alpha_1,\ldots,\alpha_N)\ | \ \alpha_i\geq r_i,\ \sum_{k\in\mathcal{N}}\alpha_k=C_{\mathcal{N}}\right\}
$$
is a constrained correlated equilibrium. Moreover any probability distribution over the maximal face of the capacity region $\mbox{max-face}(\mathcal{C})$ is a correlated constrained equilibrium distribution.
\end{coro}

\section{Hybrid AWGN Multiple Access Control}

In this section, we extend the single receiver case to one with multiple receivers. Multi-input and multi-output (MIMO) channel access game has been studied in the context of power allocation and control. For instance,  the authors in \cite{cioffi2003} formulate a two-player zero-sum game where the first player is the group of transmitters and the second one is the set of MIMO sub-channels. In \cite{belmega2010}, the authors formulate an $N$-person non-cooperative power allocation game and study its equilibrium under two different decoding schemes.

Here, we formulate a hybrid multiple access game
where users are allowed to select their rates and channels under capacity constraints. We first obtain general results on the existence of the equilibrium and methods to characterize it. In addition,  we investigate long-term behavior of the strategies and apply evolutionary game dynamics to both rates and channel selection probabilities. We show that G-function based dynamics is appropriate for our hybrid model by viewing the channel selection probabilities as strategies that determine of fitness of rate selection. Using the generalized Smith dynamics for channel selection, we are able to build an overall hybrid evolutionary dynamics for
the static model. Based on simulations, we confirm  the validity of these proposed dynamics and the correspondence between the rest point of the dynamics and the Nash equilibrium.


\subsection{Hybrid Model with Rate Control and Channel Selection}\label{HybridRateModel}

In this subsection, we establish a model for multiple users and multiple receivers.  Each user needs to decide the rate at which to transmit  and the channel to pick. We formulate a game $\overline{\Xi}=\langle\mathcal{N}, (\mathcal{A}_i)_{i\in\mathcal{N}}, (\overline{U}_i)_{i\in\mathcal{N}}\rangle$, in which the decision variable is $(\alpha_i,\mathbf{p}_{i})$, and $\mathbf{p}_i=[p_{ij}]_{j\in\mathcal{J}}$ is a $J$-dimensional vector, where $p_{ij}$ is the probability of user $i\in\mathcal{N}$ to choose channel $j\in\mathcal{J}$ and $p_{ij}$ needs to satisfy the probability measure constraints
\begin{equation}\label{probconstraint}
\sum_{j\in\mathcal{J}}p_{ij}=1, p_{ij}\geq 0, \forall i\in\mathcal{N}
\end{equation}
The game $\overline{\Xi}$ is {\it asymmetric} in the sense that the strategy sets of users are  different and the payoffs are not symmetric. 

Let $C_{j,\Omega}:=\log\left(1+\sum_{i\in \Omega}\frac{P_{ij} h_{ij}}{\sigma^2_0}\right)$ be the capacity for a subset $\Omega\subseteq\mathcal{N}$ of users at receiver $j\in\mathcal{J}$ and $r_{ij,\Omega}:=\log\left(1+\frac{P_{ij}h_{ij}}{\sigma_0^2+\sum_{i'\in \Omega,i'\neq i} P_{i'j}h_{i'j}}\right)$  the  bound rate  of a user $i$ when the signals of the $|\Omega|-1$ other users are treated as  noise at receiver $j$. Each receiver $j$ has a  capacity region $\mathcal{C}(j)$ given by
\begin{eqnarray}\label{receiverConstraint}
 \mathcal{C}(j)=\left\{(\alpha,\mathbf{p}_{j})\in [0,1]^{N}\times \mathbb{R}_+^{N} \ \bigg| \   
\sum_{i\in \mathcal{N}}\alpha_ip_{ij}\leq  C_{j,\Omega} ,
\forall \  \emptyset\subset \Omega_j \subseteq \mathcal{N},\ j\in \mathcal{J}\right\},
\end{eqnarray}
The expected payoff function $\overline{U}_i:\prod_{i=1}^N\mathcal{A}_i\longrightarrow \mathbb{R}_+$ of the game is given by
\begin{eqnarray}\label{UtilitySection3}
\overline{U}_i(\alpha_i, \mathbf{p}_i, \alpha_{-i},\mathbf{p}_{-i})=\mathbb{E}_{\mathbf{p}_i}[u_{ij}(\alpha,\mathbf{P})]=\sum_{j\in \mathcal{J}}p_{ij}u_{ij}(\alpha, \mathbf{P}),
\end{eqnarray}
where $\alpha =(\alpha_i,\alpha_{-i})\in\mathbb{R}^N_+$ and $\mathbf{P}\in[0,1]^{N\times J}=(\mathbf{p}_i,\mathbf{p}_{-i})$, with $\mathbf{p}_i\in[0,1]^J,\mathbf{p}_{-i}\in[0,1]^{(N-1) \times J}$. Assume that the utility $u_{ij}$ of a user $i$ transmitting to receiver $j$  is only dependent on the user himself and is described by  a positive and strictly increasing function $g_{i}:\mathbb{R}_+\rightarrow\mathbb{R}_+$, i.e., $u_{ij}=g_i, \forall j\in\mathcal{J},$ when capacity constraints are satisfied.

With the presence of coupled constraints (\ref{receiverConstraint}) from each receiver and probability measure constraint (\ref{probconstraint}), each sender has his individual optimization problem (IOP) given as follows.
\begin{eqnarray}
\nonumber \max_{\mathbf{p}_i,\alpha_i}& \sum_{j\in\mathcal{J}}p_{ij}g_{i}(\alpha_ip_{ij})\\
\nonumber \textrm{s.t.}& \sum_{j\in\mathcal{J}}p_{ij}=1, \forall i\in\mathcal{N}\\
\nonumber & p_{ij}\geq 0, \forall i\in\mathcal{N},j\in\mathcal{J}\\
\nonumber & (\alpha, \mathbf{p}_j)\in\mathcal{C}(j) ,\forall j\in\mathcal{J}
\end{eqnarray}

Denote the feasible set of (IOP) by $\mathcal{F}=\mathcal{F}_1\times\mathcal{F}_2$, where 
\begin{equation}\label{F1Section3}
\mathcal{F}_1=\left\{\alpha\in\mathbb{R}^N_+ \mid (\alpha, \mathbf{P})\in\cap_{j\in\mathcal{J}}\mathcal{C}(j), \mathbf{P}\in \mathcal{F}_2\right\}, 
\end{equation}
\begin{equation}\label{F2Section3}
\mathcal{F}_2=\left\{\mathbf{P}\in\mathbb{R}^{N\times J} \mid \sum_{j\in\mathcal{J}}p_{ij}=1, p_{ij}\geq0, \forall i\in\mathcal{N}, j\in\mathcal{J} \right\}.
\end{equation}
The action set of each user can thus be described by \begin{equation}
\mathcal{A}_i(\alpha_{-i},\mathbf{p}_{-i})=\left\{(\alpha_i,\alpha_{-i})\in \mathcal{F}_1, (\mathbf{p}_i,\mathbf{p}_{-i})\in\mathcal{F}_2\right\}.
\end{equation}

\subsubsection{An Example}
Suppose we have three users and three receivers, 
that is, $\mathcal{N}=\{1,2,3\}$ and $\mathcal{J}=\{1,2,3\}$.
The capacity region at receiver $1$ is given by
$$\mathcal{C}(1)=\left\{\begin{array}{c}\alpha_{i}\geq 0,\  i=1,2,3\\
p_{11}\alpha_{1}\leq \log(1+\frac{P_1h_1}{\sigma_0^2})\\
p_{21}\alpha_{2}\leq \log(1+\frac{P_2h_2}{\sigma_0^2})\\
p_{31}\alpha_{3}\leq \log(1+\frac{P_3h_3}{\sigma_0^2})\\
p_{11}\alpha_{1}+p_{21}\alpha_{2}\leq \log(1+\frac{P_1h_1+P_2h_2}{\sigma_0^2})\\
p_{11}\alpha_{1}+p_{31}\alpha_{3}\leq \log(1+\frac{P_1h_1+P_2h_2}{\sigma_0^2})\\
p_{21}\alpha_{2}+p_{31}\alpha_{3}\leq \log(1+\frac{P_1h_1+P_2h_2}{\sigma_0^2})\\
p_{11}\alpha_{1}+p_{21}\alpha_{2}+p_{31}\alpha_{3}\leq \log(1+\frac{P_1h_1+P_2h_2+P_3h_3}{\sigma_0^2})\\
0\leq p_{i1}\leq 1,\ i=1,2,3\\
\end{array}\right\}.$$
This can be written into
\begin{eqnarray*}
\mathcal{C}(1)= \left\{  \mathbf{p}_{1}=
\left[
\begin{array}{c}
  p_{11}  \\
  p_{21}      \\
  p_{31}
\end{array}
\right]
\in [0,1]^{3},
\left[\begin{array}{c}\alpha_1\\ \alpha_2\\ \alpha_3\end{array}\right]\in\mathbb{R}_+^3\ \bigg|
M_3\left[\begin{array}{c}p_{11}\alpha_1\\ p_{21}\alpha_2\\ p_{31}\alpha_{3}\end{array}\right]\leq \left[\begin{array}{c}C_{1,\{1\}}\\ C_{1,\{2\}}\\ C_{1,\{3\}}\\ C_{1,\{1,2\}}\\ C_{1,\{1,3\}}\\ C_{1,\{2,3\}}\\ C_{1,\{1,2,3\}}\end{array}\right]\right\},\end{eqnarray*}
where $C_{1,\Omega}= \log\left(1+\sum_{i\in\Omega}\frac{P_{i1}h_{i1}}{\sigma_0^2}\right)$ and $M_3$ is a totally unimodular matrix: $M_3:=
\left[
\begin{array}{ccc}
 1 & 0  & 0   \\
 0 & 1   & 0   \\
  0 &  0 &  1 \\
1 & 1  & 0   \\
1 & 0   & 1   \\
  0 &  1 &  1 \\
  1 &  1 &  1
\end{array}
\right].$ Capacity regions at receivers 2 and 3 can be obtained in a similar way.

\subsection{Characterization of Constrained Nash Equilibria}\label{characterization}

In this subsection, we characterize the Nash equilibria of the defined game $\overline{\Xi}$ under the given capacity constraint. We use the following theorem to prove the existence of Nash equilibrium for the case where the rates are predetermined;  this result is then used to solve the game for the case when both the rates and the connection probabilities are (joint) decision variables.
\begin{thm}\label{exist}
(Ba\c{s}ar \& Olsder, \cite{basar95}) Let $\mathcal{A}=\mathcal{A}_1\times \mathcal{A}_2\cdots \times \mathcal{A}_N$ be a closed, bounded and convex subset of $\mathbb{R}^N$, and for each $i\in\mathcal{N}$, the payoff functional $\overline{U}_i:\mathcal{A}\rightarrow \mathbb{R}$ be jointly continuous in $\mathcal{A}$ and concave in $a_i$ for every $a_j \in \mathcal{A}_j, j\in \mathcal{N}, j\neq i$. Then, the associated $N$-person nonzero-sum game admits a Nash equilibrium in pure strategies.
\end{thm}

Applying Theorem \ref{exist}, we have the following results immediately.

\begin{prop}\label{exist1}
Suppose $\alpha_i, i\in\mathcal{N},$ are predetermined feasible rates. Let feasible set $\mathcal{F}$ be closed, bounded and convex. If $g_{i}$ in (IOP) are continuous on $\mathcal{F}$ and concave in $\mathbf{p}_i$ (without the assumption of being positive and strictly increasing), the expected payoff functions $\overline{U}_i:
\mathbb{R}^N_+\times[0,1]^{N\times J}\rightarrow \mathbb{R}$ are concave in $\mathbf{p}_i$ and continuous on $\mathcal{F}$, then the static game admits a Nash equilibrium.
\end{prop}

The existence result in Proposition \ref{exist1} only captures the case where the rates $\alpha_i$ are predetermined, and relies on the convexity requirement of the utility functions $g_i$. We can actually obtain a stronger existence result by observing that the formulated game $\overline{\Xi}$ is a potential game with a potential function given by
\begin{equation}
\Psi(\alpha, \mathbf{P})=\sum_if_i(\alpha_i,\mathbf{p}_i)=\sum_i\sum_jp_{ij}g_i(\alpha_ip_{ij}),
\end{equation}
where $f_i=\sum_{j}p_{ij}g_i(\alpha_ip_{ij})$, the expected payoff to user $i$.

Note that the feasible set $\mathcal{F}$ is generally nonempty and bounded.
We can conclude the existence result in Proposition  \ref{exist2} of NE from this observation.

\begin{prop}\label{exist2}
The hybrid rate control game $\overline{\Xi}$ admits a Nash equilibrium.
\end{prop}

\begin{proof}
Let us formulate a centralized optimization problem (COP) as follows.
$$\begin{array}{ccc}
&\max_{\alpha,\mathbf{P}} & \Psi (\alpha, \mathbf{P})  \\
&\textrm{s.t.} & {(\alpha,\mathbf{P})}\in\mathcal{F}
\end{array}$$
Using the result in \cite{Zhu08}, we can conclude that if there exists a solution to (COP), then there exists a Nash equilibrium to the game $\overline{\Xi}$. Since $\mathcal{F}$ is compact and nonempty, and the
objective function is continuous, there exists a solution to (COP) and thus a Nash equilibrium to the game.
\end{proof}

The problem above is generally not convex and uniqueness of the Nash equilibrium may not be guaranteed. However, we
still can further characterize the Nash equilibrium through the following propositions. 

\begin{prop}\label{bestResponseProp}
Let $\beta_{ij}:=\alpha_ip_{ij}$. Without predetermining $\alpha$, suppose that $(\mathbf{p}_{-i},\alpha_{-i})$
is feasible. A best response strategy at receiver $j\in\mathcal{J}$ for user $i$ must satisfy
\begin{equation}\label{IneqBestResp}
0\leq p_{ij}\alpha_i\leq C_{j,\Omega_j}-\sum_{k\neq i}\alpha_kp_{kj}, \forall \Omega_j
\end{equation}
where $\Omega_j:=\{\Omega\in 2^{\mathcal{N}} \mid i'\in \Omega,  p_{i'j}>0 \}$ is the set of users transmitting to receiver $j.$
Since $g_{i}$ is a non-decreasing function, the best correspondence at $j$ is
{\small
\begin{equation}\label{BestResponseExp}
\beta_{ij}^*=\alpha_i^*p_{ij}^*=\max\left(r_{ij,\mathcal{N}},\min_{\Omega_j} \left(C_{j,\Omega_j}{-\sum_{i'\neq i}\alpha_{i'}p_{i'j}}\right)\right),\end{equation}}
where $r_{ij,\mathcal{N}}$ is the bound on the rate of  user $i$  when signals of $|{\mathcal{N}}|-1$ other users are treated as noise.
\end{prop}
\begin{proof} The proof is immediate by observing that the rate of user $i$ at receiver $j$ must satisfy (\ref{IneqBestResp}) due to the coupled constraints. Thus, the maximum rate that user $i$ can use to transmit to receiver $j$ without violating the constraints is clearly the minimum of $ C_{j,{\Omega}_j}-\sum_{i'\neq i}\alpha_{i'} p_{i'j}$ over all $\Omega_j$. Since the payoff is a non-decreasing function, the best response for $i$ at receiver $j$ is given by (\ref{BestResponseExp}).\end{proof}

\begin{prop}\label{Kprop}
Let $K^*_i=\textrm{arg}\max_{j\in\mathcal{J}} g_{ij}(\beta_{ij})$. If $K^*_i=\{k^*\}$ is a singleton, then the best reply for user $i$ is to choose $$\left\{\begin{array}{cc}p_{ij}=1& \textrm{if~} j=k^*_i, \\
p_{ij}=0 & \textrm{otherwise,}\end{array}\right. $$
and we can determine $\alpha_i$ by
$\alpha_i=\frac{\beta_{ik^*}}{p_{ik^*}}$.\\
If $|K^*_i|\geq 2$, then the best response correspondence is
$$\left\{\begin{array}{ll}\mathbf{p}_i\in \Delta (K^*_i)& \textrm{if~}j\in K^*_i, \\ 0 & \textrm{otherwise.~}\end{array}\right. $$
We can determine $\alpha_i$ from $\beta_{ij}$ by
$\alpha_i=\sum_{j\in K^*_i}\beta_{ij}.$
\end{prop}

\begin{proof}
Since the expected utility is given in the form  of
$$U_i(\alpha_i,\mathbf{p}_i,\alpha_{-i},\mathbf{p}_{-i})=\mathbb{E}_{\mathbf{p}_i}[u_{ij}(\alpha_{i}p_{ij})],$$
the expected utility under best response is $U_i=\mathbb{E}_{\mathbf{p}_i}[u_{ij}(\beta_{ij})].$
If for a singleton $k^*$ such that $k^*=\textrm{arg}\max_{i\in\mathcal{N}}g_{ij}(\beta_{ij}),$ we can assign all the weight $p_{ik^*}=1$ to maximize the expected utility. Since $\beta_{ik^*}=\alpha_ip_{ik^*}$, then $\alpha=\beta_{ik^*}/p_{ik^*}=\beta_{ik^*}.$ If the set $K^*$ is not a singleton,
without loss of generality, we can pick two indices $j$,$j'\in K^*$ such that $\beta_{ij}=\alpha_ip_{ij}$ and $\beta_{ij'}=\alpha_ip_{ij'}$, leading to $u_{ij}(\beta_{ij})=u_{ij'}(\beta_{ij'})$. Since the utilities to transmit using $j$ and $j'$ are the same, we can assign arbitrary (two-point) distribution, $p_{ij}$ and $p_{ij'}$ over them, with $p_{ij}+p_{ij'}=1$. Therefore, $\beta_{ij}+\beta_{ij'}=\alpha_i(p_{ij}+p_{ij'})=\alpha_i.$
\end{proof}

\subsection{Multiple Access Evolutionary Games }\label{evolutionaryGames}
Interactions among users are dynamic and the users can update their rates and channel selection with  respect to their payoffs and the known coupled constraints. Such a dynamic process can generally be modeled by either an evolutionary process, a learning process or a trial-and-error updating process.
 In classical game theory, the focus is on strategies that optimize payoffs to the players while, in evolutionary game theory, the focus is on strategies that will persist through time. 
In this subsection, we formulate evolutionary games dynamics based on the static game discussed in Section \ref{HybridRateModel}. We use generalized Smith-dynamics for channel selection and G-function based dynamics for rates. Combining them, we set up a framework of hybrid dynamics for the overall system.

The action of each user has two components $(\alpha_i,\mathbf{p}_i)\in \mathbb{R}_+\times [0,1]^J$. We use $\mathbf{p}_i$ as strategies that determine the fitness of user $i$'s rate $\alpha_i$ to receiver $j$. 
The rate selection evolves according to the channel selection strategy  $\mathbf{P}$. We may view channel selection as an inner game that involves a process on a short time scale but the rate selection is an outer game that represents the dynamical link via fitness on a longer time scale, \cite{vincent}, \cite{vincent05}.

\subsubsection{Learning the Weight Placed on Receiver }
Let $\alpha$ be a fixed  rate on the capacity region. We assume that user $i$ occasionally
experiments the weights $p_{ij}$ with alternative receivers, keeping the new strategy if and only if it leads to
a strict increase in payoff.  If the choice of receivers' weights of some users
 decreases the payoff or violates the constraints due to a strategy change
by another user, he starts a random search for a new strategy, eventually settling on one with a
probability that increases monotonically with its realized payoff. For the above generating function based dynamics, the weight of switching from receiver $j$ to receiver $j'$  is given  by
$$\eta_{jj'}^i(\alpha,\mathbf{P})=\max(0, u_{ij'}(\alpha,\mathbf{P})-u_{ij}(\alpha,\mathbf{P}))^{\theta},\ \theta\geq 1$$ if the payoff obtained at receiver $j'$ is greater the payoff obtained receiver $j$ and the constraints are satisfied; otherwise, $\eta_{jj'}^i(p,\alpha)=0.$ The frequencies of uses of each receiver is then seen as the selection strategy of receivers.

The expected change at each receiver is the difference between the incoming flow and the outgoing flow. The dynamics is also called {\it generalized Smith dynamics} \cite{cdc}  and is given by
{\small
\begin{eqnarray}\label{smith84}
\dot{p}_{ij}(t)=\sum_{j'\in\mathcal{J}}p_{ij'}(t)\eta_{j'j}^i(\alpha,\mathbf{P}(t))-
p_{ij}(t)\sum_{j'\in\mathcal{J}}\eta_{jj'}^i(\alpha,\mathbf{P}(t)).
\end{eqnarray}
}
Let $\chi_{ij}(\alpha,\mathbf{P}(t)):=\sum_{j'\in\mathcal{J}}p_{ij'}(t)\eta_{j'j}^i(\alpha,\mathbf{P}(t))-
p_{ij}(t)\sum_{j'\in\mathcal{J}}\eta_{jj'}^i(\alpha,\mathbf{P}(t)).$ Hence, the dynamics can be rewritten as $\dot{p}_{ij}=\chi_{ij}(\alpha,\mathbf{P}(t))$. For $\theta=1$ the dynamics is known as {\it Smith dynamics} and  has been used for describing
the evolution of road traffic congestion in which the fitness is determined by the strategies
chosen by all drivers. It has also been studied in the context of the resource selection in hybrid systems and migration constraint
problem in wireless networks in \cite{cdc}.
 \begin{prop} \label{propc}
 Any equilibrium of the game $\overline{\Xi}$ with predetermined rates is a rest points of the generalized Smith dynamics (\ref{smith84}).
 \end{prop}
 \begin{proof}
The transition rate between receivers preserves the sign in the sense that, for every user,
incoming flow  from the receiver $j' $  to $j$ is positive if and only if the constraints are satisfied and the payoff to $j$ exceeds the payoff to $j'.$ Let $\alpha$ be a feasible point. We first remark that if the right hand side of (\ref{smith84}) is non-zero for some splitting strategy $\mathbf{P},$ then
\begin{eqnarray} \nonumber
d &:= &\sum_{j\in\mathcal{J}}\ \dot{p}_{ij} u_{ij}(\alpha,\mathbf{P})=\sum_{j\in\mathcal{J}}\ \chi_{ij} u_{ij}(\alpha,\mathbf{P})\\ \nonumber
 & = & \sum_{j,j'\in\mathcal{J}}  p_{ij'}\left(u_{ij}(\alpha,\mathbf{P})-u_{ij'}(\alpha,\mathbf{P}) \right)\eta^i_{j'j}\\ \nonumber
 &=& \sum_{j,j'\in\mathcal{J}}  p_{ij'}\max\left[0,\left(u_{ij}(\alpha,\mathbf{P})-u_{ij'}(\alpha,\mathbf{P}) \right)\right]\eta^i_{j'j}
\end{eqnarray} which is strictly positive. Thus, if $(\alpha,\mathbf{P})$ is a Nash equilibrium then $(\alpha,\mathbf{P})$ satisfy the constraints, and $p_{ij}=0,$ or $\eta^i_{jj'}(\alpha,\mathbf{P})=0.$ This implies that $(\alpha,\mathbf{P})$ satisfies also $\chi(\alpha,\mathbf{P}))=0.$
\end{proof}

 The following proposition says that the equilibria are exactly the  rest point of (\ref{smith84}).
 \begin{prop}\label{restpoint}
 Any rest point of the dynamics (\ref{smith84}) is a Nash  equilibrium of the game $\overline{\Xi}$.
 \end{prop}

 The proof of  Proposition  \ref{restpoint} can be obtained by using Theorem III in \cite{cdc}.
 Since the probability to switch from receiver $j$ to $j'$ is proportional to $\eta^i_{jj'}$, which preserves the sign of payoff difference, we can use the Theorems III in \cite{cdc}. It follows that the dynamics generated by $\eta$ satisfy the Nash stationarity property.

\subsubsection{G-function Based Dynamics}
We introduce here the generating fitness function ({\it G-function}) based dynamics with projection onto the capacity region. The G-function approach has been successfully applied to non-linear continuous games by Vincent and Brown~\cite{vincent}, \cite{vincent05}. It is appropriate for our hybrid model because we can regard the channel selection as the variables in a fitness function. Users choose channel selection probabilities to aim at increasing their fitness of their rate choice. In our rate allocation game, to deal with constraints, we use projection into capacity region in order to preserve the  trajectories feasible.
Starting from a point in the polytope $\mathcal{C}$, each user revises and updates its strategy according to a rate proportional to the gradient and its payoff subject to the capacity constraints. Let $G_{ij}$ be the fitness generating function of user $i$ at receiver $j$ defined on $\mathbb{R}^N\times\mathbb{R}^{N\times J}$ satisfying
\begin{equation}
\nonumber G_{ij}(v,\alpha, \mathbf{P}){\bigg|_{v=\mathbf{p}_i}}=\left(C_{j,\mathcal{N}}-p_{ij}\beta_{ij}(t)-\sum_{i'\in\mathcal{N}\backslash\{i\}} p_{i'j}\beta_{i'j}(t)\right), \end{equation}
if $(\alpha,\mathbf{P})$ satisfies in the hybrid capacity region. Notice that the term~$C_{j,\mathcal{N}}-\sum_{i'\neq i} p_{i'j}\beta_{i'j}(t)$ is maximum rate of $i$ using channel $j$ at time $t$.
Hence, the G-function based dynamics is given by
\begin{equation}\label{GDynamics}
\dot{\beta}_{ij}= -\bar{\mu}\left[p_{ij}\beta_{ij}-C_{j,\mathcal{N}}+\sum_{i'\neq i} p_{i'j}\beta_{i'j}\right]p_{ij}\beta_{ij}.
\end{equation} 
with initial conditions $\beta_{ij}(0)\leq C_{j,\{i\}}$, where $\beta=[\beta_{ij}]$ is defined in Proposition \ref{bestResponseProp}, which is of the same dimension as $\alpha$, and $\alpha_i(t)=\sum_{j\in\mathcal{J}}\beta_{ij}(t)$;  $\bar{\mu}$ is an appropriate parameter chosen for the rate of convergence.

\subsubsection{Hybrid Dynamics}
We now combine the two evolutionary game dynamics described in the previous subsections.  Variables $\alpha$ and $\mathbf{P}$ are both evolving in time. The overall dynamics are given by

\begin{equation}\label{HybridDynamics}
\left\{
\begin{array}{lll}
\dot{p}_{ij}(t)&=&\sum_{j'\in\mathcal{J}}p_{ij'}(t)\eta_{j'j}^i(\alpha(t),\mathbf{P}(t))
-p_{ij}(t)
\sum_{j'\in\mathcal{J}}\eta_{jj'}^i(\alpha(t),\mathbf{P}(t))
\\
\dot{\beta}_{ij}(t)&=& -\bar{\mu}\left[p_{ij}(t)\beta_{ij}(t)-C_{j,\mathcal{N}}
+ \sum_{i'\neq i} p_{i'j}(t)\beta_{i'j}(t)\right]p_{ij}(t)\beta_{ij}(t)\medskip\\ \alpha_i(t)&=& \sum_{j\in\mathcal{J}}\beta_{ij}(t)\medskip,\ \beta_{ij}(0)\leq C_{j,\{i\}}, \forall j\in\mathcal{J},\ i\in\mathcal{N}
\end{array}\right.
\end{equation}
All the equilibria of the hybrid evolutionary rate control and channel selection are  rest point of the above hybrid dynamics. The following result can be obtained directly from Proposition \ref{propc} and (\ref{GDynamics}).

\begin{prop} Let $(\beta^*,\mathbf{P}^*)$ be an interior rest points of the hybrid dynamics, i.e., $\beta_{ij}^*>0,\ p^*_{ij}>0$ and $\chi(\alpha^*,\mathbf{P})=0.$ Then for all $j$, $$\sum_{i=1}^N p^*_{ij}\beta^*_{ij}= C_{j,\mathcal{N}}; ~~ \chi\left(\sum_{j=1}^N\beta^*_{ij},\mathbf{P}^*\right)=0.$$
\end{prop}

\subsection{Numerical Examples}
In this subsection, we illustrate the evolutionary dynamics in (\ref{GDynamics}) and (\ref{HybridDynamics}) by examining a two-user and three-receiver communication system as depicted in Figure \ref{2user3Rx}. Let $h_{i1}=0.1, h_{i2}=0.2, h_{i3}=0.3$, for $i=\{1,2\}$. Each transmission power $P_i$ is set to $1 $ mW for all $i=1,2$ and the noise level is set to $\sigma^2=-20$ dBm.

\begin{figure}[htb]
\begin{center}
  \includegraphics[scale=0.5]{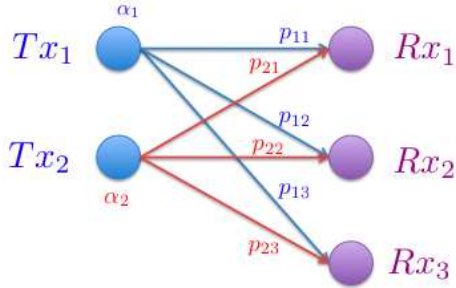}\\
  \caption{Two users and three receivers}\label{2user3Rx}
\end{center}
\end{figure}

In the first experiment, we assume that the rates of the users are predetermined to be $\alpha=[10, 20]^T$, the Smith dynamics in (\ref{GDynamics}) yield in Figure (\ref{FixedAlphap1})  and (\ref{FixedAlphap2}) the response of $\mathbf{p}_1$ and $\mathbf{p}_2$. It can be seen that the dynamics converge very fast within less than half a second.

In the second experiment, we assume that the probability matrix $\mathbf{P}$ has been optimally found by the users using   (\ref{smith84}). Figures \ref{beta1} and \ref{beta2} show that the $\beta$ values converge to an equilibrium from which we can find the optimal value for $\alpha$. Since these dynamics are much slower compared to Smith dynamics on $\mathbf{P}$, our assumption of knowledge of optimal $\mathbf{P}$ for a slowly varying $\alpha$ becomes valid.

In the next experiment, we simulate the hybrid dynamics in (\ref{HybridDynamics}). Let the probability $p_{ij} $ of user $i$ choosing transmitter $j$ and the transmission rates be initialized as follows:
$$\begin{array}{ll}
\mathbf{P}(0)=\left[\begin{array}{ccc}
0.2 & 0.3& 0.5 \\
    0.25 &0.5& 0.25 \end{array}\right],
& \alpha(0)=\left[\begin{array}{c}0.2 \\ 0.1\end{array}\right].
\end{array}$$
We let the parameter $\bar{\mu}=0.9$. Figure \ref{p1} shows the evolution of the weights of user 1 on each of the receivers. The weights converge to be $p_{1j}=1/3$ for all $j$ within two seconds, leading to an unbiased choice among receivers. In Figure \ref{p2}, we show the evolution of the weights of the second user on each receiver. At the equilibrium, $\mathbf{p}_{2}=[0.3484, 0.4847, 0.1669]^T$. It appears that user 2 favors the second transmitter over the other ones. Since the utility $u_{ij}$ is of the same form, the optimal response set $K_i^*$ is naturally nonempty and contains all the receivers. Shown in Proposition \ref{Kprop}, the probability of choosing a receiver  at the equilibrium is randomized among the three receivers and can be determined by the rates $\alpha$ chosen by the users.

The $\beta$-dynamics determines the evolution of $\alpha$ in  (\ref{HybridDynamics}). In Figure \ref{alpha}, we see that the evolutionary dynamics yield $\alpha=[15.87, 23.19]^T$ at the equilibrium. It is easy to verify that they satisfy the capacity constraints outlined in Section \ref{secmodel}. It converges within 5 seconds and appears be much slower than in Figures \ref{p1} and \ref{p2}. Hence, it can be seen that $\mathbf{P}-$dynamics may be seen as the inner loop dynamics while $\beta-$dynamics can be seen as an outer loop evolutionary dynamics. They evolve on two different time scales. In addition, thanks to Proposition \ref{restpoint},  finding the rest points for the above dynamics ensures us finding the equilibrium.
\begin{figure*}[h!tb]
\begin{minipage}[b]{0.4\linewidth} 
\centering
\includegraphics[scale=0.4]{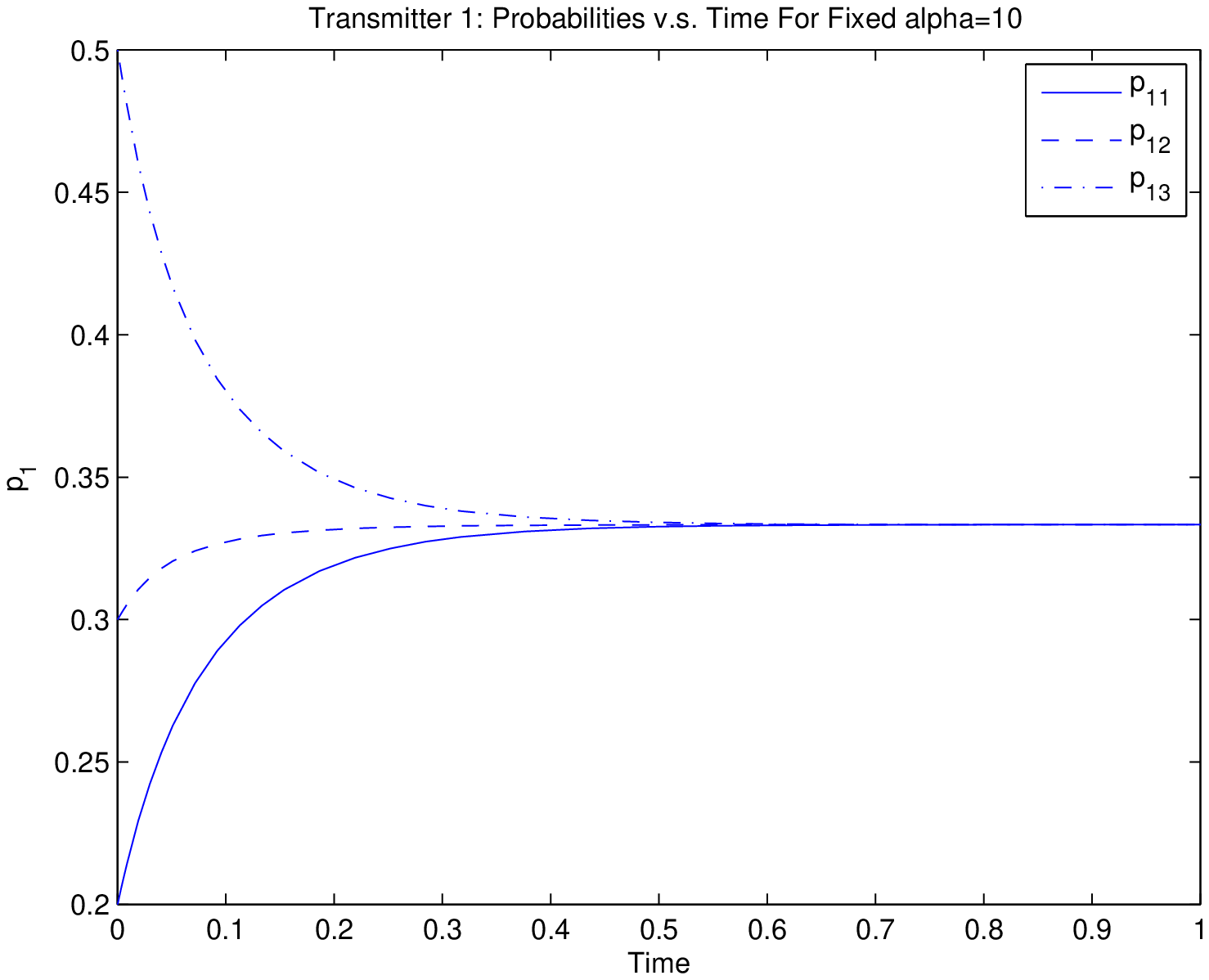}\caption{Transmitter 1: Probabilities v.s. Time For Fixed $\alpha_1$=10} \label{FixedAlphap1}
\end{minipage}
\hspace{0.3cm}
\begin{minipage}[b]{0.4\linewidth}
\centering
\includegraphics[scale=0.4]{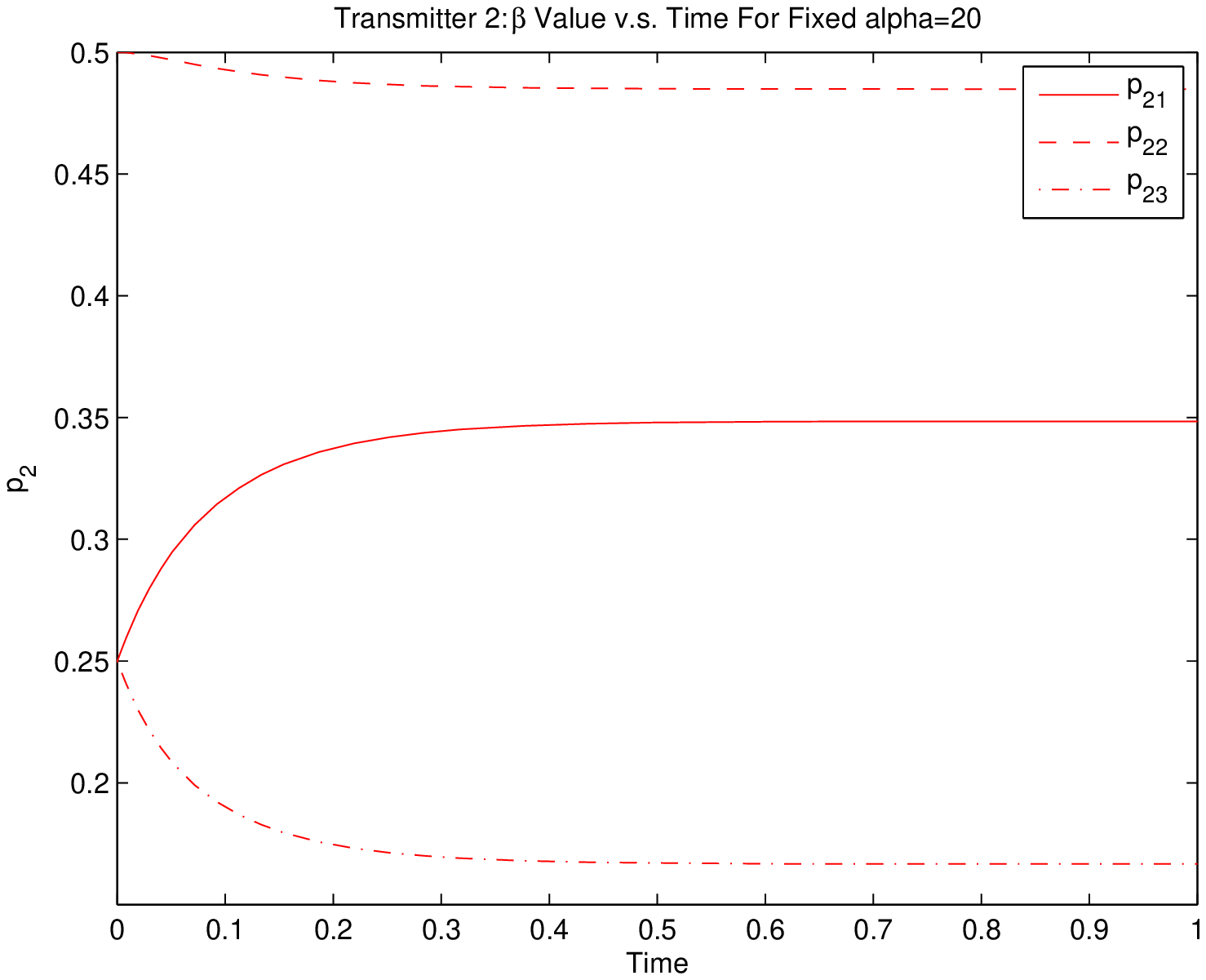}
\caption{Transmitter 2: Probabilities v.s. Time For Fixed $\alpha_2$=20} \label{FixedAlphap2}
\end{minipage}
\end{figure*}

\begin{figure*}[h!tb]
\begin{minipage}[b]{0.4\linewidth} 
\centering
\includegraphics[scale=0.4]{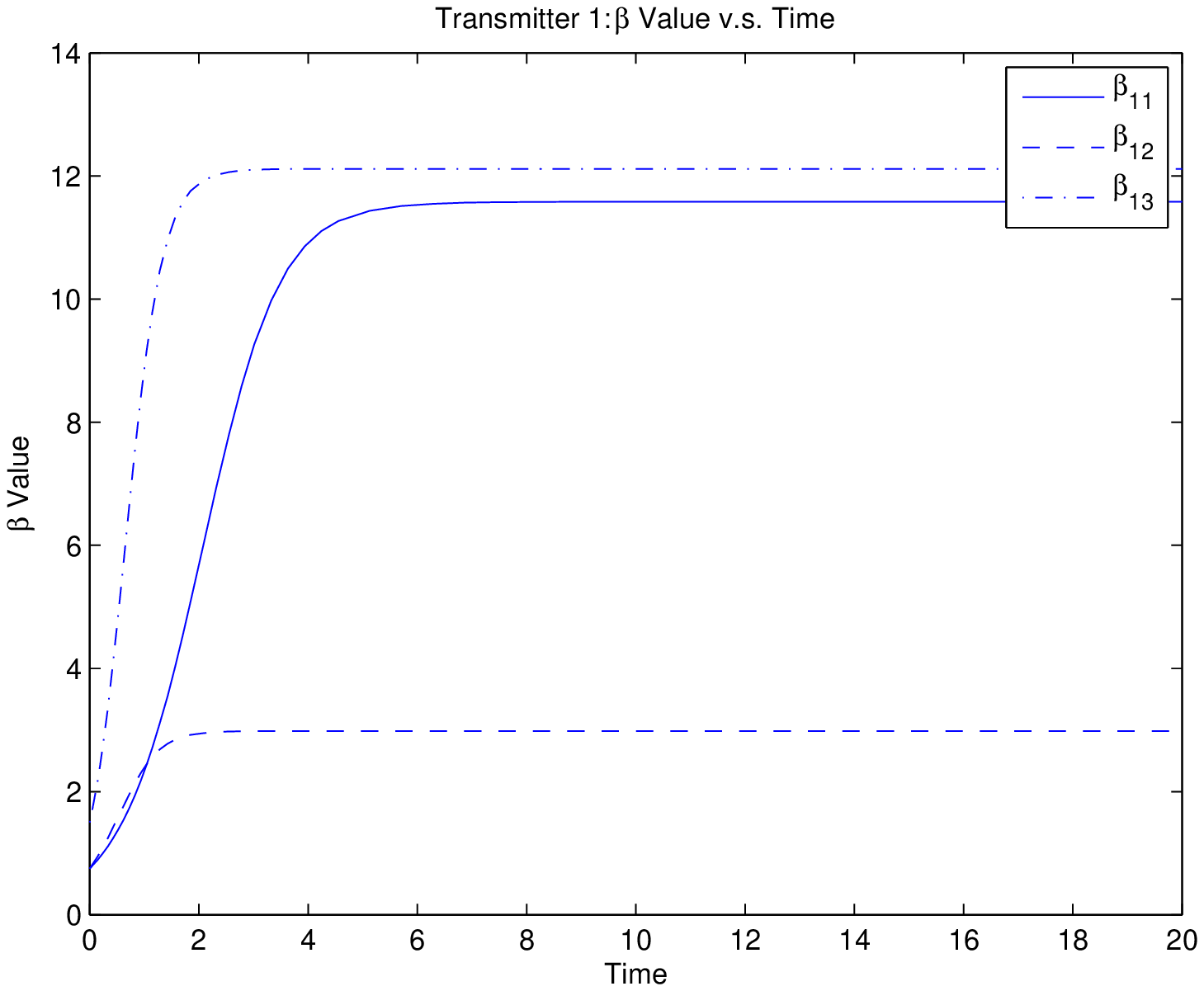}\caption{Transmitter 1: $\beta$ Value v.s. Time} \label{beta1}
\end{minipage}
\hspace{0.3cm} 
\begin{minipage}[b]{0.4\linewidth}
\centering
\includegraphics[scale=0.4]{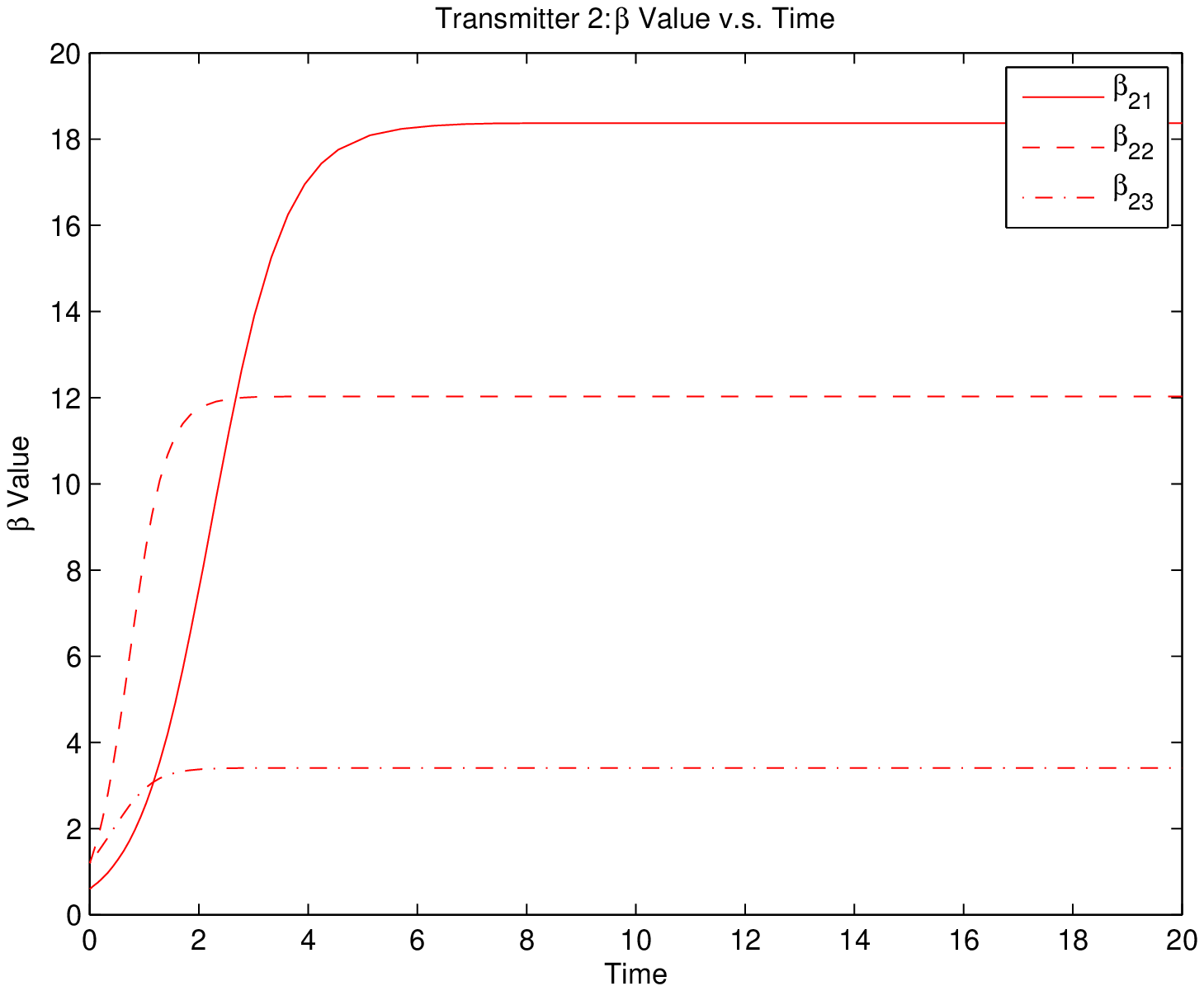}
\caption{Transmitter 1: $\beta$ Value v.s. Time} \label{beta2}
\end{minipage}
\end{figure*}

\begin{figure*}[h!tb]
\begin{minipage}[b]{0.3\linewidth} 
\centering
\includegraphics[scale=0.3]{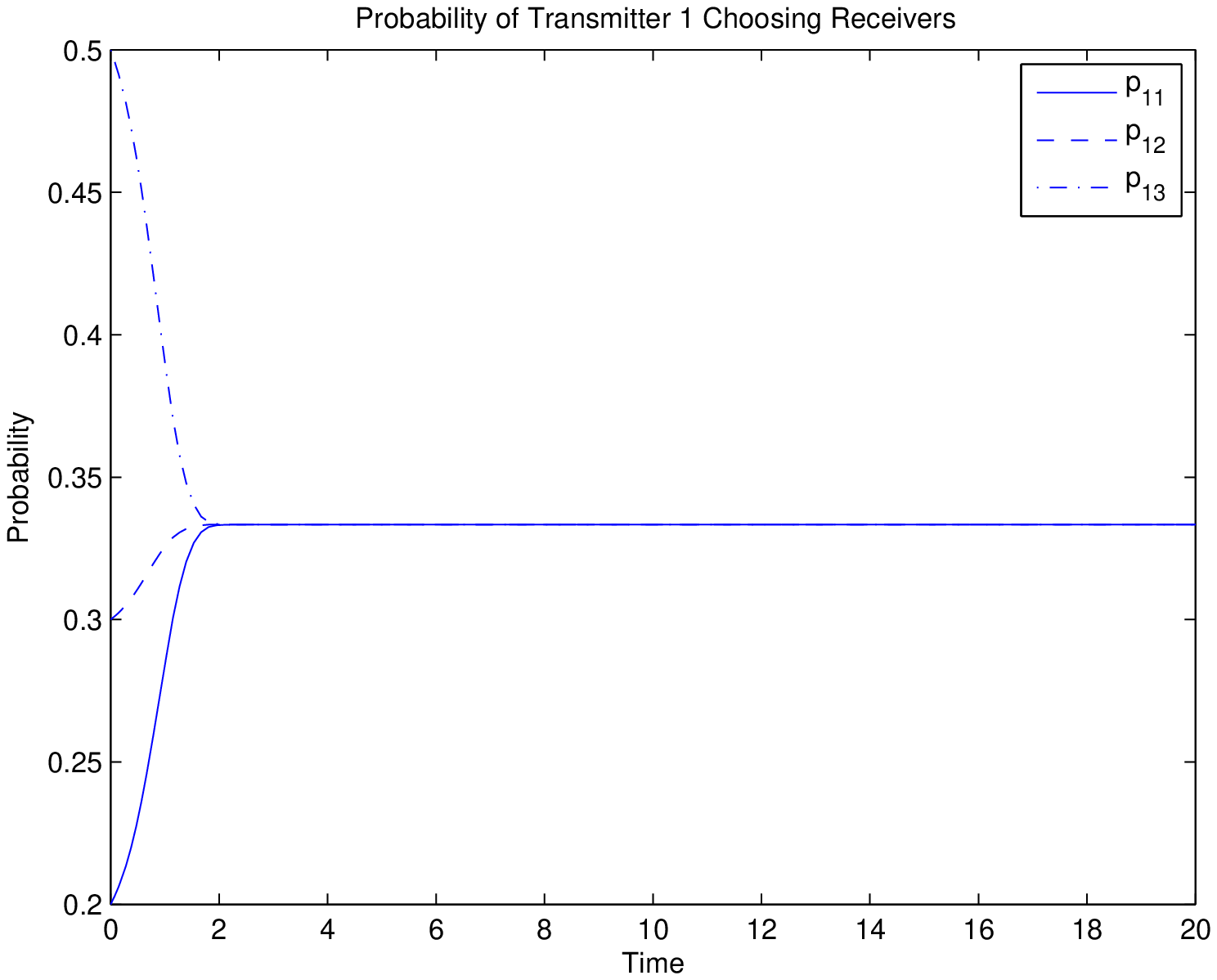} \caption{Probability of Transmitter 1 Choosing Receivers} \label{p1}
\end{minipage}
\begin{minipage}[b]{0.3\linewidth}
\centering
\includegraphics[scale=0.3]{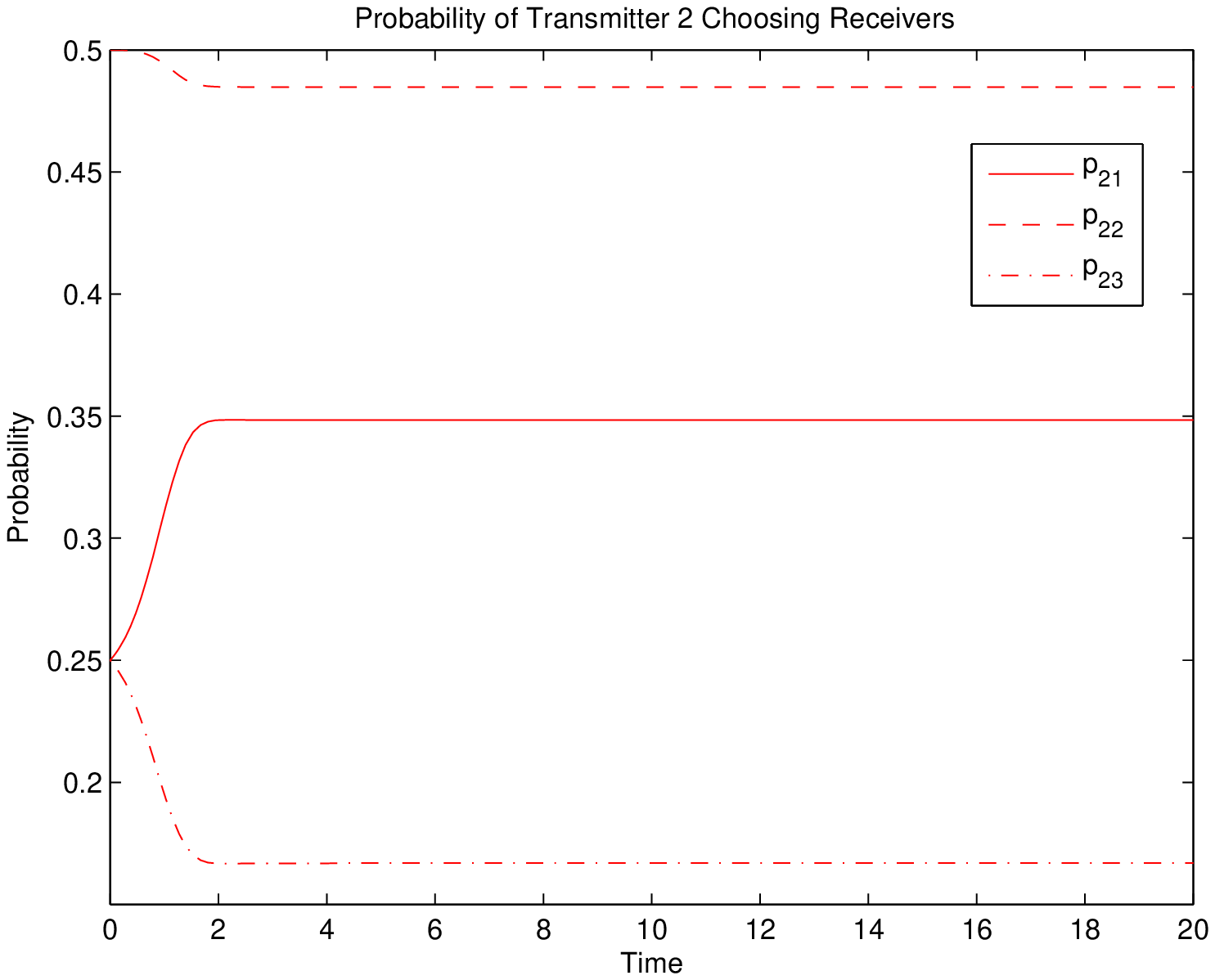}
\caption{Probability of Transmitter 2 Choosing Receivers} \label{p2}
\end{minipage}
\begin{minipage}[b]{0.3\linewidth}
\includegraphics[scale=0.3]{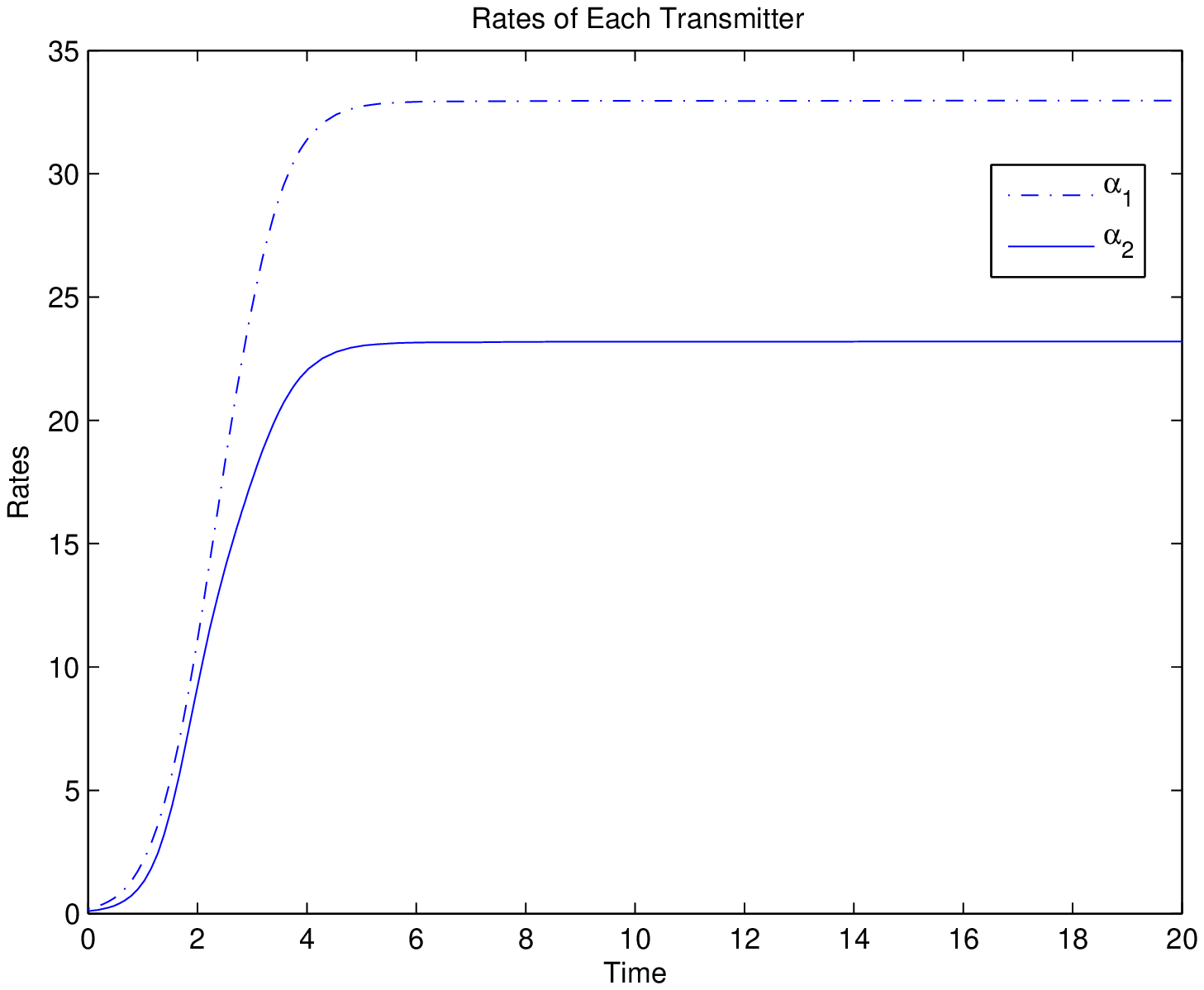}
\caption{Rates of Each Transmitter} \label{alpha}
\end{minipage}
\end{figure*}

\section{Concluding Remarks} \label{secconclud}
 In this paper, we have studied an evolutionary  multiple access channel game with a continuum action space and coupled rate constraints. We showed that the  game has a continuum of strong equilibria which  are 100\% efficient in the rate optimization problem.  We proposed the constrained Brown-von Neumann-Nash dynamics, Smith dynamics, and the replicator dynamics  to study the stability of equilibria in the  long run. In addition, we have introduced a hybrid multiple access game model and its corresponding evolutionary game-theoretic framework. We have analyzed the Nash equilibrium for the static game and suggested a system of evolutionary game dynamics based method to find it. It is found that the Smith dynamics for channel selections are a lot faster than the $\beta$-dynamics, and the combined dynamics yield a rest point that corresponds to the Nash equilibrium. An interesting extension that we leave for future  research is to introduce a dynamic channel characteristics: the gains $h_{ij}(t)$ are {\it time-dependent random variables}. Another interesting question is
to find equilibria structure  in the case of multiple access games with {\it non-convex capacity regions}.


\end{document}